\documentclass[11pt]{article}
\usepackage{amsmath,amssymb,amsfonts}
\usepackage{color,graphicx}
\usepackage{enumerate,latexsym}
\usepackage[sort,nocompress,space]{cite}
\usepackage{exscale,cmmib57}
\newtheorem{theorem}{Theorem}

\newtheorem{remark}[theorem]{\textit{Remark}}

\newenvironment{proof}[1][Proof]{\noindent\textbf{#1.} }{\ \rule{0.5em}{0.5em}}
\newenvironment{AMS}[1][AMS Subject Classification]{\noindent\textbf{#1:} }{}
\newenvironment{DOI}[1][DOI]{\noindent\textbf{#1:} }{}
\newenvironment{keywords}[1][Keywords]{\noindent\textbf{#1:} }{}

\def\trp{^T}
\def\Cbb{\mathbb{C}}
\def\Rbb{\mathbb{R}}
\def\Ncal{\mathcal{N}}
\def\diag{{\rm diag}}
\def\half{\frac{1}{2}}

\newcommand{\dst}{\displaystyle}

\title{Network Synthesis of Linear Dynamical Quantum Stochastic
Systems\thanks{Journal version published electronically September 18, 2009. URL: http://www.siam.org/journals/sicon/48-4/72865.html.}}

\author{Hendra I. Nurdin\thanks{Department of Information Engineering, Australian National University,
Canberra, ACT 0200, Australia (Hendra.Nurdin@anu.edu.au, Matthew.James@anu.edu.au). The
research of these authors was supported by the Australian Research Council.}
\and Matthew R. James\footnotemark[2]
\and Andrew C. Doherty\thanks{School of Physical Sciences, University of Queensland,
Queensland 4072, Australia (doherty@ physics.uq.edu.au).}}

\begin{document}
\maketitle

\begin{abstract}
The purpose of this paper is to develop a synthesis theory for linear dynamical quantum
stochastic systems that are encountered in linear quantum optics and in phenomenological
models of linear quantum circuits. In particular, such a theory will enable the systematic
realization of coherent/fully quantum linear stochastic controllers for quantum control,
amongst other potential applications. We show how general linear dynamical quantum
stochastic systems can be constructed by assembling an appropriate interconnection of
one degree of freedom open quantum harmonic oscillators and, in the quantum optics setting,
discuss how such a network of oscillators can be approximately synthesized or implemented in
a systematic way from some linear and nonlinear quantum optical elements. An example is also
provided to illustrate the theory.
\end{abstract}


\begin{AMS}
93B50, 93B10, 93E03, 94C99, 81V80
\end{AMS}

\begin{DOI}
10.1137/080728652
\end{DOI}

\pagestyle{myheadings}
\thispagestyle{plain}
\markboth{H. I. NURDIN, M. R. JAMES, AND A. C. DOHERTY}{LINEAR QUANTUM STOCHASTIC
NETWORK SYNTHESIS}

\section{Background and motivation}\label{sec:introduction}
In recent years there has been an explosion of interest in exploitation of quantum mechanical
systems as a basis for new quantum technologies, giving birth to the field of quantum information
science. To develop quantum technologies, it has been recognized from early on that quantum
control systems will play a crucial role for tasks such as manipulating a quantum mechanical
system to perform a desired function or to protect it from external disturbances \cite{Feyn59,DM03}.
Moreover, recent advances in quantum and nanotechnology have provided a great impetus for research
in the area of quantum feedback control systems; e.g., see
\cite{VPB83,HW94a,DJ99,AASDM02,SGDM04,Mab08}.

\looseness=-1Perhaps just about the simplest and most tractable controller to design would be the linear
quantum controllers, and this makes them an especially attractive class of controllers. In this
class, one can have classical linear quantum controllers that  process only classical signals
which are obtained from a quantum plant by measurement of some plant output signals (e.g.,
\cite{Bel79,DJ99,EB05}), but more recently there has also been interest in fully quantum and
mixed quantum-classical linear controllers that are able to manipulate quantum signals
\cite{YK03b,JNP06,Nurd07,NJP07a,NJP07b}. In fact, an experimental realization of a fully
quantum controller in quantum optics has been successfully demonstrated in \cite{Mab08}. As noted
in that paper, the class of fully quantum controllers or {\em coherent-feedback controllers}, as they
are often known in the physics literature, presents genuinely new control-theoretic challenges
for quantum controller design. An important open problem raised in the works
\cite{JNP06,Nurd07,NJP07a,NJP07b} is how one would systematically build or implement a general,
arbitrarily complex, linear quantum controller, at least approximately, from basic quantum devices,
such as quantum optical devices. This problem can be viewed as a quantum analogue of the synthesis
problem of classical electrical networks (in this paper the qualifier ``classical'' refers broadly
to systems that are not quantum mechanical) that asks the question of how to build arbitrarily
complex linear electrical circuits from elementary passive and active electrical components such
as resistors, capacitors, inductors, transistors, op-amps, etc. Therefore, the quantum synthesis
problem is not only of interest for the construction of linear quantum stochastic controllers, but
also more broadly as a fundamental aspect of linear quantum circuit theory that arises, for example,
in quantum optics and when working with phenomenological models of quantum RLC circuits such as
described in \cite{YD84}, as well as in relatively new fields such as nanomechanical circuit
quantum electrodynamics \cite{TC06,SG08}.

A key result of this paper is a new synthesis theorem (Theorem \ref{th:synthesis}) that prescribes
how an arbitrarily complex linear quantum stochastic system can be decomposed into an interconnection
of basic building blocks of one degree of freedom open quantum harmonic oscillators and thus be
systematically constructed from these building blocks. In the context of quantum optics, we then
propose physical schemes for ``wiring up'' one degree of freedom open quantum harmonic oscillators
and the interconnections between them that are required to build a desired linear quantum stochastic
system, using basic quantum optical components such as optical cavities, beam splitters, squeezers,
etc. An explicit yet simple example that illustrates the application of the theorem to the synthesis
of a two degrees of freedom open quantum harmonic oscillator is provided.

\subsection{Elements of linear electrical network synthesis}\label{sec:background}
To motivate synthesis theory in the context of linear dynamical quantum systems, we start with a brief
overview of aspects of linear electrical network synthesis that are relevant for the current work.

As is well known, a classical (continuous time, causal, linear time invariant) electrical network
described by a set of (coupled) ordinary differential equations can be analyzed using various
representations, for example, with a frequency domain or transfer function representation, with
a modern  state space representation and, more recently, with a behavioral representation. It is
well known that the transfer function and state space representation are equivalent in the sense
that one can switch between one representation to the other for any given network. However,
although one can associate a unique transfer function representation to a state space representation,
the converse is not true: for a given transfer function there are infinitely many state space
representations. The state space representation can be made to be unique (up to a similarity
transformation of the state space matrices) by requiring that the representation be of minimal
order (i.e., the representation is both controllable and observable). The synthesis question in
linear electrical networks theory deals with the inverse scenario, where one is presented with
a transfer function or state space description of a linear system and would like to synthesize
or build such a system from various linear electrical components such as resistors, capacitors,
inductors, op-amps, etc. A particularly advantageous feature of the state space representation,
since it is given by a set of first order ordinary differential equations, is that it can be
inferred directly from the representation how the system can be {\em systematically} synthesized.
For example, consider the system below, given in a state space representation:
\begin{eqnarray}
\label{eq:ss-eg-1} \frac{dx(t)}{dt} &=&\left[\begin{array}{cc} 2 &
5 \\ -2 & -4
\end{array}\right]x(t)+ \left[\begin{array}{c} 1\\ 0.1
\end{array}\right]u(t),\\
y(t)&=&\left[\begin{array}{cc} 0 & 1 \end{array}\right]x(t) +
u(t),  \nonumber
\end{eqnarray}
where $x(t)$ is the state, $u(t)$ is the input signal, and $y(t)$ is the output signal. In an
electrical circuit, $u(t)$ could be the voltage at certain input ports of the circuit and $y(t)$
\begin{figure}[tbp!]
\centerline{\includegraphics[width=17pc]{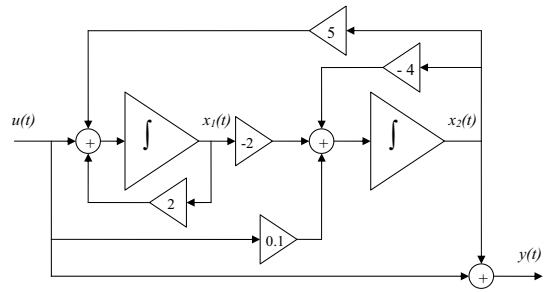}}
\caption{Schematic for the implementation of the classical system {\rm (\ref{eq:ss-eg-1})}.\label{fig:circ-diag}}
\end{figure}
\begin{figure}[tbp!]
\centerline{\includegraphics[width=19pc]{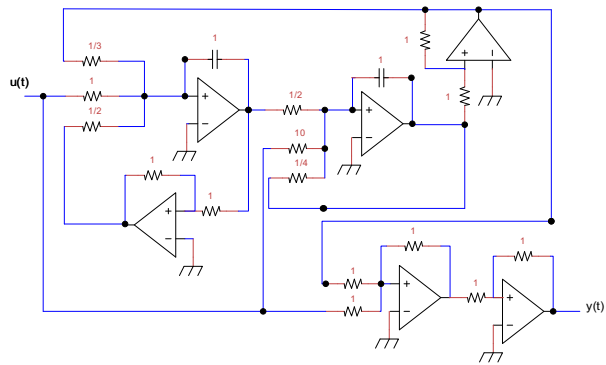}}
\caption{Hardware implementation of the schematic diagram shown in Figure
{\rm\ref{fig:circ-diag}}.\label{fig:circ-hw}}
\end{figure}
could be the voltage at another set of ports of the circuit, different from the input ports. This
system can be implemented according to the schematic shown in Figure \ref{fig:circ-diag}. This
schematic can then be used to to implement the system at the hardware level as shown in
Figure \ref{fig:circ-hw} \cite[Chapter 13]{AV73}. However, linear electrical network synthesis
is a mature subject that deals with much more than just how one can obtain {\em some} realization
of a particular system. For instance, it also addresses fundamental issues such as how a passive
network, a network that does not require an external source of energy,  can also be synthesized
using only passive electrical components, and how to synthesize a given circuit with a minimal
number of circuit elements or with a minimal number of certain types of elements (such as active
elements). In this paper our primary objective is to develop an analogously systematic method for
synthesizing arbitrarily complex linear dynamical {\em quantum} stochastic systems that are given
in an abstract description that is similar in form to (\ref{eq:ss-eg-1}). These linear dynamical
quantum stochastic systems are ubiquitous in linear quantum optics, where they arise as idealized
models for linear open quantum systems. However, since there is currently no comprehensive synthesis
theory available for linear dynamical quantum systems (as opposed to {\em static} linear quantum
systems in linear quantum optics that have been studied in, e.g., \cite{Leon03}) and related notions
such as passivity have not been extensively studied and developed, here we focus our attention
solely on the development of a {\em general} synthesis method that applies to {\em arbitrary} linear
dynamical quantum systems which does not exploit specific physical properties or characteristics
that a particular system may possess (say, for instance, passivity). Although the latter will be an
important issue to be dealt with in further development of the general theory, it is beyond the
scope of the present paper (which simply demonstrates the existence of {\em some} physical
realization).

\subsection{Open quantum systems and quantum Markov models}\label{sec1.2}
\looseness=-1A quantum system is never completely isolated from its environment and can thus interact with it.
Such quantum systems are said to be {\em open quantum systems} and are important in modeling
various important physical phenomena such as the decay of the energy of an atom. The environment
is modeled as a separate quantum system in itself and can be viewed as a {\em heat bath} to
which an open quantum system can dissipate energy or from which it can gain energy (see
\cite[Chapters 3 and 7]{GZ00}). An idealization often employed in modeling the interaction
between an open quantum system and an external heat bath is the introduction of a {\em Markovian}
assumption: the dynamics of the coupled system and bath is essentially ``memoryless'' in the sense
that future evolution of the dynamics of the coupled system  depends only on its present
state and not at all on its past states. Open quantum systems with such a property are said to
be {\em Markov}. The Markov assumption is approximately valid under some physical assumptions
made on the system and bath, such as that the heat bath is so much ``larger'' than the system
(in the sense that it has many more degrees of freedom than the system) and is weakly coupled
to the system that its interaction with the latter has little effect on its own dynamics and
can thus be neglected; for details on the physical basis for this Markovian assumption, see
\cite[Chapters 3 and 5]{GZ00}. Markov open quantum systems are important, as they are often
employed as very good approximations to various practically relevant open quantum systems,
particularly those that are encountered in the field of quantum optics, yet at the same time
are relatively more tractable to analyze as their dynamics can be written in terms of first
order operator differential equations.

In Markov open quantum systems, heat baths can be idealistically modeled as a collection of
a continuum of harmonic oscillators oscillating at frequencies in a continuum of values. An
important consequence of the Markov approximation in this model is that the heat bath can be
effectively treated in a quantum statistical sense as quantum noise \cite[section 3.3]{GZ00},
and thus Markov open quantum systems have inherently stochastic quantum dynamics that are most
appropriately described by quantum stochastic differential equations
(QSDE) \cite{GZ00,HP84,KRP92,BvHJ07}.
To be concrete, a single heat bath in the Markov approximation is formally modeled as an
operator-valued {\em quantum white noise} process $\eta(t)$, where $t \geq 0$ denotes time,
that satisfies the singular commutation relation $[\eta(t),\eta(t)^*]=\delta(t-t')$, where $^*$
denotes the adjoint of an operator, $\delta(t)$ is the Dirac delta function, and the commutator
bracket $[\cdot,\cdot]$ acts on operators $A$ and $B$ as $[A,B]=AB-BA$.  Examples of heat baths
that have been effectively modeled in such a way include vacuum noise, squeezed and laser
fields in quantum optics \cite{GZ00}, and infinitely long bosonic transmission lines \cite{YD84}.
See also \cite{YK03a} for a brief intuitive overview of the modeling of a free-traveling
quantized electromagnetic wave as quantum white noise. The formal treatment with quantum white
noises can be made mathematically rigorous by considering the bosonic annihilation process $A(t)$
(on a Fock space) that can be formally defined as the ``integral'' of
$\eta(t)$: $A(t)=\int_{0}^{t}\eta(s) ds$ and its adjoint process $A^*(t)=A(t)^*$. We shall refer
to the operator process $A(t)$ simply as a {\em bosonic field}. The celebrated Hudson--Parthasarathy  (H-P)
stochastic calculus provides a framework for working with differential equations involving the
processes $A(t)$ and $A^*(t)$, as well as another fundamental process on a Fock space called the
gauge process, denoted by $\Lambda(t)$, that models scattering of the photons of the bosonic heat
bath (at a formal level, one could write $\Lambda(t)=\int_{0}^{t} \eta(s)^*\eta(s) dt$). More
generally, a quantum system can be coupled to several independent bosonic fields
$A_1(t),\ldots,A_n(t)$,  with $A_j(t)=\int_{0}^{t} \eta_j(s) ds$, and in this case there can be
scattering between different fields modeled by interfield gauge processes
$\Lambda_{jk}(t)=\int_{0}^{t} \eta_j(s)^* \eta_k(s)ds$ (in interfield scattering, a photon is
annihilated in one field and then created in another).

\subsection{Linear dynamical quantum stochastic systems}\label{sec1.3}
Linear dynamical quantum stochastic systems (e.g., see \cite{EB05,JNP06}) arise in practice
as idealized models of {\em open} quantum harmonic oscillators whose canonical position and
momentum operators are {\em linearly} coupled to one or more external (quantum) heat baths
(the mathematical modeling involved is discussed in section \ref{sec:models}). Here a quantum
harmonic oscillator is a quantized version of a classical harmonic oscillator in which the
classical position and momentum variables $q_c$ and $p_c$, respectively, are replaced by
operators $q$, $p$ on an appropriate Hilbert space (in this case the space $L^2(\Rbb)$)
satisfying the canonical commutation relations (CCR) $[q,p]=2i$. It is said to be open if
it is interacting with elements of its environment. For instance, consider the scenario in
\cite{DJ99} of an atom trapped in an optical cavity. The light in the cavity is strongly
coupled to the atomic dipole, and as the atom absorbs and emits light, there are random
mechanical forces on the atom. In an appropriate parameter regime, the details of the optical
and atomic dipole dynamics are unimportant, and the optical field can be modeled as an
environment for the atomic motion. Under the assumptions of \cite{DJ99} the ``motional
observables'' of  the trapped atom  (its position and momentum operators) can then be treated
like those of  an open quantum harmonic oscillator. Linear Markov open quantum models are
extensively employed in various branches of physics in which the Markov type of arguments and
approximations such as discussed in the preceding subsection can be justified. They are
particularly prominent in quantum optics, but have also been used, among others, in
phenomenological modeling of quantum RLC circuits \cite{YD84}, in which the dissipative heat
baths are realized by infinitely long transmission lines attached to a circuit. For this
reason, the general synthesis results developed herein (cf.\ Theorem \ref{th:synthesis}) are
anticipated to be be relevant in various branches of quantum physics that employ linear
Markov models. For example, it has the potential of playing an important role in the
systematic and practical design of complex linear photonic circuits as the technology
becomes feasible.

A general linear dynamical quantum stochastic system is simply a many degrees of freedom
open quantum harmonic oscillator with several pairs of canonical position and momentum operators
$q_k,p_k$, with $k$ ranging from 1 to $n$, where $n$ is the number of degrees of freedom of the
system, satisfying the (many degrees of freedom) CCR $[q_j,p_k]=2i\delta_{jk}$ and
$[q_j,q_k]=[p_j,p_k]=0$, where $\delta_{jk}$ is the Kronecker delta which takes on the value
$0$ unless $j=k$, in which case it takes on the value 1, that is linearly coupled to a number of
external bosonic fields $A_1,\ldots,A_m$. In the interaction picture with respect to the field
and oscillator dynamics, the operators $q_j,p_j$ evolve unitarily in time as $q_j(t),p_j(t)$
while preserving the CCR $[q_j(t),p_k(t)]=2i\delta_{jk}$ and $[q_j(t),q_k(t)]=[p_j(t),p_k(t)]=0$
$\forall t \geq 0$, and the dynamics of the oscillator is given by (here
$x(t)=(q_1(t),p_1(t),\ldots,q_n(t),p_n(t))^T$ and $A(t)=(A_1(t),\ldots,A_m(t))^T$)
\begin{eqnarray}
dx(t)&=&A x(t) dt + B \left[ \begin{array}{c} dA(t) \\ dA(t)^*\end{array}\right],\nonumber\\
dy(t)&=&C x(t) dt + D dA(t), \label{eq:q-ss}
\end{eqnarray}

\noindent where $A \in \Rbb^{n \times n}$, $B \in \Cbb^{n \times 2m}$, $C \in \Cbb^{m \times n}$, and
$D \in \Cbb^{m \times m}$. Here the variable $y(t)$ acts as the output of the system due to
interaction of the bosonic fields with the oscillator; a component $y_j(t)$ of $y(t)$ is the
transformed version of the field $A_j(t)$ that results {\em after} it interacts with the
oscillator. Hence, $A_j(t)$ can be viewed as an {\em incoming} or {\em input} field, while
$y_j(t)$ is the corresponding {\em outgoing} or {\em output} field. To make the discussion more
concrete, let us consider a well-known example of a linear quantum stochastic system in quantum
optics: an optical cavity (see section \ref{sec:opt-cav} for further details of this device),
shown in Figure \ref{fig:fabry-perot}. The cavity depicted in the picture is known as a standing
\begin{figure}[tbp!]
\centerline{\includegraphics[width=10.5pc]{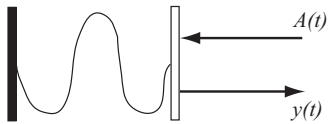}}
\caption{A Fabry--Perot optical cavity. The black rectangle denotes a mirror fully reflecting at
the cavity resonance frequency, while the white rectangle denotes a mirror partially transmitting
at that frequency.\label{fig:fabry-perot}}
\end{figure}
wave or Fabry--Perot cavity and consists of one fully reflecting mirror at the cavity resonance
frequency and one partially transmitting mirror. Light that is trapped inside the cavity forms a
standing wave with an oscillation frequency of $\omega_{cav}$, while parts of it leak through
the partially transmitting mirror. The loss of light through this mirror is modeled as an
interaction between the cavity with an incoming bosonic field $A(t)$ in the vacuum state
(i.e., a field with zero photons or a zero-point field) incident on the mirror. The dynamics
for a cavity is linear and given by
\begin{eqnarray*}
dx(t)&=&-\frac{\gamma}{2} x(t) dt -\sqrt{\gamma}dA(t),\\
dy(t)&=&\sqrt{\gamma} x(t) dt dt + dA(t),
\end{eqnarray*}
where $\gamma>0$ is the coupling coefficient of the mirror, $x(t)=(q(t),p(t))^T$ are the
interaction picture position and momentum operators of the standing wave inside the cavity, and
$y(t)$ is the outgoing bosonic field that leaks out of the cavity. A crucial point to notice
about (\ref{eq:q-ss}) is that it is in a similar form to the classical deterministic state space
representation such as given in (\ref{eq:ss-eg-1}), with the critical exception that
(\ref{eq:q-ss}) is a (quantum) stochastic system (due to the quantum statistical interpretation
of $A(t)$) and involves quantities which are operator-valued rather than real/complex-valued.
Furthermore, the system matrices $A,B,C,D$ in (\ref{eq:q-ss}) cannot take on arbitrary values for
(\ref{eq:q-ss}) to represent the dynamics of a physically meaningful system (see \cite{JNP06} and
\cite[Chapter 7]{Nurd07} for further details). For instance, for arbitrary choices of $A,B,C,D$ the
many degrees of freedom CCR may not be satisfied for all $t \geq 0$ as required by quantum mechanics;
hence these matrices cannot represent a physically feasible system. In \cite{JNP06,Nurd07}, the
notion of {\em physically realizable} linear quantum stochastic systems has been introduced that
corresponds to open quantum harmonic oscillators (hence are physically meaningful), which do not
include scattering processes among the bosonic fields. In particular, necessary and sufficient
conditions have been derived on the matrices $A,B,C,D$ for a system of the form (\ref{eq:q-ss})
to be physically realizable. More generally, however, are linear quantum stochastic systems that
are completely described and parameterized by three (operator-valued) parameters: its Hamiltonian
$H=\half x^T R x$ ($R \in \mathbb{R}^{n \times n}$, $R=R^T$), its linear coupling operator to the
external bosonic fields $L=K x$ ($K \in \Cbb^{m \times n}$), and its {\em unitary} scattering matrix
$S \in \Cbb^{m \times m}$. In particular, when there is no scattering involved ($S=I)$, then it has
been shown in \cite{JNP06} that $(S,L,H)$ can be recovered from $(A,B,C,D)$ (since $S=I$, here
necessarily $D=I$) and vice-versa. Although \cite{JNP06} does not consider the scattering processes,
the methods and results therein can be adapted accordingly to account for these processes (this is
developed in section \ref{sec:param-corspnds} of this paper).

The works \cite{JNP06,Nurd07} were motivated by the problem of the design of robust fully quantum
controllers and left open the question of how to systematically build arbitrary linear quantum
stochastic controllers as a suitable network of basic quantum devices. This paper addresses this
open problem by developing synthesis results for general linear quantum stochastic systems
for applications that are anticipated to extend beyond fully quantum controller synthesis, and it
also proposes how to implement the synthesis in quantum optics. The organization of the rest of
this paper is as follows. Section \ref{sec:models} details the mathematical modeling of linear
dynamical quantum stochastic systems and defines the notion of an open oscillator and a generalized
open oscillator, section \ref{sec:concat-ser-red-net} gives an overview of the notions of the
concatenation and series product for generalized open oscillators as well as the concept of a
reducible quantum network with respect to the series product, and section \ref{sec:param-corspnds}
discusses the bijective correspondence between two descriptions of a linear dynamical quantum
stochastic system. This is then followed by section \ref{sec:syn-theory} that develops the main
synthesis theorem which shows how to decompose an arbitrarily complex linear dynamical quantum
stochastic system as an interconnection of simpler one degree of freedom generalized open oscillators,
section \ref{sec:system-synth} that proposes the physical implementation of arbitrary one degree of
freedom generalized open oscillators and direct interaction Hamiltonians between these oscillators,
and section \ref{sec:example} that provides an explicit example of the application of the main
synthesis theorem to the construction of a two degrees of freedom open oscillator. Finally,
section \ref{sec:conclude} provides a summary of the contributions of the paper and conclusions.

\section{Mathematical modeling of linear dynamical quantum stochastic systems}\label{sec:models}
In the previous works \cite{EB05,JNP06} linear dynamical quantum stochastic systems were essentially
considered as open quantum harmonic oscillators. Here we shall consider a more general class of
linear dynamical quantum stochastic systems consisting of the cascade of a static passive linear
quantum network with an open quantum harmonic oscillator. However, in this paper we restrict our
attention to synthesis of linear systems with purely quantum dynamics, whereas the earlier work
\cite{JNP06} considers a more general scenario where a mixture of both quantum and classical
dynamics is allowed (via the concept of an augmentation of a quantum linear stochastic system).
The class of mixed classical and quantum controllers will be considered in a separate work. To
this end, let us first recall the definition of an open quantum harmonic oscillator (for further
details, see \cite{JNP06,EB05,Nurd07}).

In this paper we shall use the following notations: $i=\sqrt{-1}$, $^*$ will denote the adjoint of
a linear operator as well as the conjugate of a complex number, if $A=[a_{jk}]$ is a matrix of
linear operators or complex numbers, then $A^{\#}=[a_{jk}^*]$, and $A^{\dag}$ is defined as
$A^{\dag}=(A^{\#})^T$, where $^T$ denotes matrix transposition. We also define $\Re\{A\}=(A+A^{\#})/2$
and $\Im\{A\}=(A-A^{\#})/2i$ and denote the identity matrix by $I$ whenever its size can be inferred
from context and use $I_{n \times n}$ to denote an $n \times n$ identity matrix.

Let $q_1,p_1,q_2,p_2,\ldots,q_n,p_n$ be the canonical position and momentum operators, satisfying
the canonical commutation relations $[q_j,p_k]=2i\delta_{jk},\,[q_j,q_k]=0,\,[p_j,p_k]=0$ of a
{\em quantum harmonic oscillator} with a quadratic Hamiltonian $H=\half x_0^T R x_0$
($x_0=(q_1,p_1,\ldots,q_n,p_n)^T$), where $R=R^T \in \Rbb$. The integer $n$ will be referred to as
the {\em degrees of freedom} of the oscillator. Digressing briefly from the main theme of this
section, let us first discuss why the matrix $R$ can be taken to be real (symmetric) rather than
complex (Hermitian). Consider a general quadratic Hamiltonian $H$  of the form
$H=\half \sum_{j=1}^{n} (\alpha_{j} q_j^2+ \beta_{j}q_j p_j+\gamma_{j} p_j q_j+\epsilon_j p_j^2)
+  \sum_{j=1}^{n-1} \sum_{k=j+1}^n  \kappa_{jk} q_j p_k$, with
$\alpha_j,\epsilon_j,\kappa_{jk} \in \Rbb$, $\beta_j,\gamma_j \in \Cbb$, and
$\beta_j^*=\gamma_j\ \forall\ j,k$, since $H$  must be self-adjoint. Using the commutation relations
for the canonical operators, we can then write $H=\half \sum_{j=1}^{n} (\alpha_{j} q_j^2 +
\Re \{ \beta_{j} \} (q_j p_j + p_j q_j) + \epsilon_j p_j^2) + \sum_{j=1}^{n-1} \sum_{k=j+1}^n
\kappa_{jk} q_j p_k + c = \half x_0^T R x_0 + c$ for some real symmetric matrix $R$ and a real
number $c=-2\sum_{j=1}^{n}\Im\{\beta_j\}$. Since $c$  contributes only a phase factor $e^{ic}$
that has no effect on the dynamics of the oscillator,  as $e^{i(H-c)t}x_0e^{-i(H-c)t}
=e^{iHt}x_0e^{-iHt}\,\forall t \geq 0$, we may as well just discard $c$ and take $H$ to be
$H=\half x_0^T R x_0$ (i.e., the original $H$ without the  constant term). Returning to the main
discussion, let $\eta_1,\ldots,\eta_m$ be independent vacuum quantum white noise processes satisfying
the commutation relations $[\eta_j(t),\eta_k(t')^*]=\delta_{jk}\delta(t-t')$ and
$[\eta_j(t),\eta_k(t')]=0$ $\forall j,k$ and $\forall t,t' \geq 0$, and define
$A_j(t)=\int_{0}^t \eta_j(s)ds$ ($j=1,\ldots,m$) to be vacuum bosonic fields satisfying the quantum
Ito multiplication rules \cite{HP84,KRP92}
$$
dA_j(t)dA_{k}^*(t)=\delta_{jk} dt,\quad dA_j^*(t)dA_{k}(t)=0,\quad dA_j(t)dA_k(t)=0,\quad dA_j^*(t)dA_k^*(t)=0,
$$
with all other remaining second order products between $dA_j$, $dA_k^*$ and $dt$ vanishing. An
{\em open quantum harmonic oscillator}, or simply an {\em open oscillator}, is defined as a
quantum harmonic oscillator coupled to $A(t)$ via the formal time-varying {\em idealized} interaction
Hamiltonian \cite[Chapter 11]{GZ00}
\begin{equation}
H_{Int}(t)=i(L^T \eta(t)^*-L^{\dag}\eta(t)), \label{eq:ideal-int}
\end{equation}
where $L$ is a {\em linear} coupling operator given by $L=K x_0$ with $K \in \Cbb^{m \times n}$ and
$\eta=(\eta_1,\ldots,\eta_m)^T$. Although the Hamiltonian is formal since the $\eta_j(t)$'s are singular
quantum white noise processes, it can be given a rigorous interpretation in terms of Markov limits (e.g.,
\cite{Goug06a},\cite[Chapter 11]{GZ00}). The evolution of the open oscillator is then governed by the
unitary process $\{U(t)\}_{t \geq 0}$ satisfying the QSDE
\cite{EB05,JNP06,Fagno90,GZ00}
\begin{eqnarray}
dU(t)&{=}& \left(-iHdt+dA(t)^{\dag}L-L^{\dag} dA(t)-
\frac{1}{2}L^{\dag}L dt \right)U(t);\qquad U(0)=I. \label{eq:qsde-1}
\end{eqnarray}
The time-evolved canonical operators are given by $x(t)=U(t)^*x_0 U(t)$ and satisfy the QSDE
\begin{eqnarray*}
dx(t)&=&2\Theta(R+\Im\{K^{\dag}K\})x(t)dt+ 2i\Theta
[\begin{array}{cc} -K^{\dag} &
K^T\end{array}]\left[\begin{array}{c} dA(t)
\\ dA(t)^{\#} \end{array}\right],\\
 x(0)&=&x_0,
\end{eqnarray*}
where $\Theta$ is a canonical commutation matrix of the form $\Theta=\diag(J,J,\ldots,J)$, with
$$
J=\left[\begin{array}{cc}0 & 1
\\-1 & 0 \end{array}\right],
$$
while the output bosonic fields $Y(t)=(Y_1(t),\ldots,Y_n(t))^T$ that result from interaction of
$A(t)$ with the harmonic oscillator are given by $Y(t)=U(t)^{*}A(t)U(t)$ and satisfy the QSDE
\begin{eqnarray}
dY(t)&=& K x(t)dt+ dA(t).\label{eq2.3}
\end{eqnarray}
Note that the dynamics of $x(t)$ and $Y(t)$ are linear.

The input $A(t)$ of an open oscillator can first be passed through a static passive linear (quantum)
network (for example, a static passive linear optical network. See, e.g., \cite{LN04,Leon03} for
details) without affecting the linearity of the overall system dynamics; this is shown in
Figure \ref{fig:gen-osc}. Such an operation effects the transformation $A(t) \mapsto \tilde
A(t)=SA(t)$, where $S \in \Cbb^{m \times m}$ is a complex unitary matrix (i.e., $S^{\dag}S=SS^{\dag}=I$).
Thus $\tilde A(t)$ will be a new set of vacuum noise fields satisfying the same Ito rule as $A(t)$.

\begin{figure}[tbp!]
\centerline{\includegraphics[width=15pc]{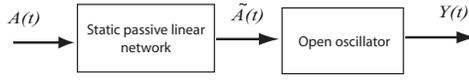}}
\caption{A generalized open oscillator.\label{fig:gen-osc}}
\end{figure}

Letting $S=[S_{jk}]_{j,k=1,\ldots,m}$, it can be shown by straightforward calculations using the
quantum Ito stochastic calculus that the cascade is equivalent (in the sense that it produces the
same dynamics for $x(t)$ and the output $y(t)$ of the system) to a linear quantum system whose
dynamics is governed by a unitary process $\{\tilde U(t)\}_{t \geq 0}$ satisfying the QSDE (for
a general treatment, see \cite{GJ07})
\begin{eqnarray}
d\tilde U(t) &=&
\left(\sum_{j,k=1}^{m}(S_{jk}-\delta_{jk})d\Lambda_{jk}(t)-iHdt
+dA(t)^{\dag}L- L^{\dag}SdA(t)\right. \label{eq:qsde-2}\\
&&\quad \left.
\vphantom{\sum_{j,k=1}^{m}(S_{jk}-\delta_{jk})d\Lambda_{jk}(t)}
-\;\frac{1}{2}L^{\dag}Ldt\right)\tilde U(t);\tilde U(0)=I,\nonumber
\end{eqnarray}
where $\Lambda_{jk}(t)$ ($j,k=1,\ldots,m$) are fundamental processes, called the gauge processes,
satisfying the quantum Ito rules
$$
d\Lambda_{jk}(t)
d\Lambda_{j'k'}(t)\hspace*{-1pt}=\hspace*{-1pt}\delta_{kj'}d\Lambda_{jk'}(t),\hspace*{6pt}dA_j(t)
d\Lambda_{kl}(t)\hspace*{-1pt}=\hspace*{-1pt}\delta_{jk}dA_l(t),\hspace*{6pt}d\Lambda_{jk} dA_l(t)
^*\hspace*{-1pt}=\hspace*{-1pt}\delta_{kl}dA_j^*(t),
$$
with all other remaining second order products between $d\Lambda_{jl}(t)$ and
$dA_{j'}(t),dA_{l'}^*(t),dt$ vanishing. This yields the following dynamics for
$x(t)=\tilde U(t)^*x_0 \tilde U(t)$ and the system output
$y(t)=\tilde U(t)^*A(t)\tilde U(t)$:
\begin{eqnarray}
dx(t) &=& A x(t)dt + B\left[\begin{array}{c} dA(t) \\ dA(t)^{\#}
\end{array}\right], \label{eq:goo-dyn-1}\\
dy(t) &=& Cx(t)dt +D dA(t), \label{eq:goo-dyn-2}
\end{eqnarray}
with
\begin{eqnarray}
\label{eq:goo-coeffs} A&=& 2\Theta(R+\Im\{K^{\dag}K\}),\\
B&=& 2i\Theta [\begin{array}{cc} -K^{\dag}S &
K^TS^{\#}\end{array}], \nonumber \\
C&=& SK, \nonumber \\
D&=& S.  \nonumber
\end{eqnarray}
For convenience, in the remainder of the paper we shall refer to the cascade of a static passive
linear quantum network with an open oscillator as a {\em generalized open oscillator}.

Let $G$ be a generalized open oscillator that evolves according to the QSDE (\ref{eq:qsde-2}) with
given parameters $S$, $L=Kx_0$, and $H=\half x_0^T Rx_0$. For compactness, we shall use a shorthand
notation of \cite{GJ07} and denote such a generalized open oscillator by $G=(S,L,H)$. In the next
section we briefly recall the concatenation and series product developed in \cite{GJ07} that allows
one to systematically obtain the parameters of a generalized open oscillator built up from an
interconnection of generalized open oscillators of one degree of freedom.

\section{The concatenation and series product of generalized open oscillators and reducible
quantum networks}\label{sec:concat-ser-red-net}
In this section we will recall the formalisms of concatenation product, series product, and
reducible networks (with respect to the series product) developed in \cite{GJ07} for the
manipulation of networks of generalized open oscillators as well as more general Markov open
quantum systems.

Let $G_1= (S_1,K_1x_{1,0},\half x_{1,0}^T R_1 x_{1,0})$ and
$G_2=(S_2,K_2x_{2,0},\half x_{2,0}^T R_2x_{2,0})$ be two generalized open oscillators, where
$x_{k,0}=x_k(0)$. The concatenation product $G_1 \boxplus G_2$ of $G_1$ and $G_2$ is defined
as
$$
G_1 \boxplus G_2= \left(S_{1\boxplus 2},(K_{1}x_{1,0},
K_2x_{2,0})^T,\half x_{1,0}^T R_{1}x_{1,0}+\half x_{2,0}^T R_{2}x_{2,0}\right),
$$
where
\begin{eqnarray*}
S_{1\boxplus2}&=&\left[\begin{array}{cc} S_1 & 0 \\ 0 &
S_2\end{array}\right].
\end{eqnarray*}
It is important to note here that the possibility that $x_{1,0}=x_{2,0}$ or that some components of $x_{1,0}$
coincide with those of $x_{2,0}$ are allowed. If $G_1$ and $G_2$ are independent oscillators (i.e., the
components of $x_{1,0}$ act on a distinct Hilbert space to that of the components of $x_{2,0}$), then the
concatenation can be interpreted simply as the ``stacking'' or grouping of the variables of two
noninteracting generalized open oscillators to form a larger generalized open oscillator.

It is also possible to feed the output of a system $G_1$ to the input of system $G_2$, with the
proviso that $G_1$ and $G_2$ have the same number of input and output channels. This operation of
cascading or loading of $G_2$ onto $G_1$ is represented by the series product $G_2 \triangleleft G_1$
defined by
\begin{eqnarray*}
G_2 \triangleleft G_1 &=&\bigg(S_2S_1,K_2x_{2,0}+S_2K_1x_{1,0},\half x_{1,0}^TR_1x_{1,0} \\
&&\hspace*{11pt}+\; \half x_{2,0}^TR_2x_{2,0}+\frac{1}{2i}x_{2,0}^T(K_2^{\dag}
S_2K_1-K_2^TS_2^{\#}K_1^{\#})x_{1,0}\bigg).
\end{eqnarray*}
Note that $G_2 \triangleleft G_1$ is again a generalized open oscillator with a scattering matrix,
coupling operator, and Hamiltonian as given by the above formula.

With concatenation and series products having been defined, we now come to the important notion of
a {\em reducible network with respect to the series product} (which we shall henceforth refer to
more simply as just a {\em reducible network}) of generalized open oscillators. This network consists
of $l$ generalized open oscillators $G_k=(S_k,L_k,H_k)$, with $L_k=K_k x_{k,0}$ and
$H_k=\half x_{k,0}^T R_k x_{k,0}$, $k=1,\ldots,l$, along with the specification of a direct interaction
Hamiltonian $H^d= \sum_{j}\sum_{k=j+1}x_{j,0}^T R_{jk} x_{k,0}$ ($R_{jk}\in \Rbb^{2 \times 2}$) and a list
$\mathcal{S}=\{G_k \triangleleft G_j\}$ of series connections among generalized open oscillators $G_j$
and $G_k$, $j \neq k$, with the condition that each input and each output has at most one connection,
i.e., lists of connections such as $\{G_2\triangleleft G_1,G_3\triangleleft G_2,G_1\triangleleft
G_3\}$ are disallowed. Such a reducible network $\Ncal$ again forms a generalized open oscillator and
is denoted by $\Ncal=\{\{G_k\}_{k=1,\ldots,l},H^d,\mathcal{S}\}$. Note that if $\Ncal_0$ is a reducible
network defined as $\Ncal_0=\{\{G_k\}_{k=1,\ldots,l},0,\mathcal{S}\}=(S_0,L_0,H_0)$, then $\Ncal$, which
is $\Ncal_0$ equipped with the direct interaction Hamiltonian $H^d$, is simply given by
$\Ncal=\Ncal_0 \boxplus (0,0,H)=(S_0,L_0,H_0+H^d)$.

The notion of a reducible network was introduced in \cite{GJ07} to study networks that are free of
``algebraic loops'' such as when connections like $\{G_2\triangleleft G_1,G_3\triangleleft
G_2,G_1\triangleleft G_3\}$ are present. The theory in \cite{GJ07} is not sufficiently general to
treat networks with algebraic loops; they can be treated in the more general framework of quantum
feedback networks developed in \cite{GJ08}. Since this work is  based on \cite{GJ07}, we also restrict
our attention to reducible networks, but as we shall show in section \ref{sec:syn-theory} this is
actually sufficient to develop a network synthesis theory of linear quantum stochastic systems.
A network synthesis theory can indeed also be developed using the theory of quantum feedback networks
of \cite{GJ08}, and this has been pursued in a separate work \cite{Nurd09b}.

Two important decompositions of a generalized open oscillator based on the series product that will
be exploited in this paper are
\begin{eqnarray}
(S,L,H)&=&(I,L,H) \triangleleft (S,0,0), \label{eq:decomp-1}\\
(S,L,H)&=&(S,0,0) \triangleleft (I,S^{\dag}L,H)
\label{eq:decomp-2},
\end{eqnarray}
where $(S,0,0)$ represents a static passive linear network implementing the unitary matrix $S$.

\section{Correspondence between system matrices \boldmath $(A,B,C,D)$ and the parameters
$S,L,H$\unboldmath}\label{sec:param-corspnds}
In \cite{JNP06} it has been shown that for $S=I$, then $D=I$, and there is a bijective correspondence
between the system matrices $(A,B,C,I)$ of a physically realizable linear quantum stochastic system
\cite[section III]{JNP06} and the parameters $K,R$ of an open oscillator; see Theorem 3.4 therein
(however, note that the $B$, $C$, and $D$ matrices are defined slightly differently from here because
\cite{JNP06} expresses all equations in terms of quadratures of the bosonic fields rather than their
modes). Here we shall show that allowing for an arbitrary complex unitary scattering matrix $S$, a
bijective correspondence between the system parameters $(A,B,C,D)$ of an extended notion of a
physically realizable linear quantum stochastic system and the parameters $S,K,R$ of  a generalized
open oscillator (in particular, $D=S$) can be established. We begin by noting that we may write the
dynamics (\ref{eq:goo-dyn-2}) in the following way:
$$
y(t)=Sy'(t),
$$
with $y'(t)$ defined as
$$
dy'(t)=S^{\dag}K x(t)dt + dA(t).
$$
Then by defining $K'=S^{\dag}K$ and substituting $K=SK'$ in (\ref{eq:goo-coeffs}), we see that $x(t)$
in (\ref{eq:goo-dyn-1}), and $y'(t)$ constitutes the dynamics for the open oscillator
$(I,K'x_0,\half x_0^TRx_0)$ with system matrices given by $(A,B,S^{\dag}C,I)$. Since $D=S$ and
$(S,L,H)=(S,0,0)\triangleleft (I,K'x_0,\half x_0^TRx_0)$ (cf.\ (\ref{eq:decomp-2})), from
\cite[Theorem 3.4]{JNP06} we see that there is a bijective correspondence between $(A,B,S^{\dag}C)$
and the parameters $(K',R)$ and that one set of parameters may be uniquely recovered from the other.
Therefore, we may define a system of the form (\ref{eq:q-ss}) to be physically realizable (extending
the notion in \cite{JNP06}) if it represents the dynamics of a generalized open oscillator (this
idea already appears in \cite[Chapter 7]{Nurd07}; see Remark 7.3.8 therein). This implies that a
system (\ref{eq:q-ss}) with matrices $(A,B,C,D)$ is physically realizable if and only if $D$ is a
complex unitary matrix and $(A,B,D^{\dag}C,I)$ are the system matrices of a physically realizable
system in the sense of \cite{JNP06} (i.e., $(A,B,D^{\dag}C,I)$ are the system matrices of an open
oscillator). Therefore, we may state the following theorem.

\begin{theorem}\label{thm4.1}
There is a bijective correspondence between the system matrices $(A,B,C,D)$ and the parameters
$(S,K,R)$ of a generalized open oscillator. For given $(S,K,R)$, the corresponding system matrices
are uniquely given by {\rm (\ref{eq:goo-coeffs})}. Conversely, for given $(A,B,C,D)$, which are the
system matrices of a generalized open oscillator $G$ with parameters $(S,K,R)$, then $D$ is unitary, and
$S=D$ and $(A,B,D^{\dag} C,I)$ are the system matrices of some open oscillator
$G'=(I,K'x_0,\half x_0^T R x_0)$. The parameters $(K',R)$ of the open oscillator $G'$ is uniquely
determined from $(A,B,D^{\dag}C)$ by {\rm \cite[{\it Theorem} 3.4]{JNP06}} (by suitably adapting the
matrices $B$ and $D^{\dag}C$), from which the parameter $K$ of $G$ is then uniquely determined
as $K=DK'$.
\end{theorem}

Due to this interchangeability of the description by $(A,B,C,D)$ and by $(S,K,R)$ for a generalized
open oscillator, it does not matter with which set of parameters one works with. However, for
convenience of analysis in the remainder of the paper we shall work exclusively with the
parameters $(S,K,R)$.

\section{Main synthesis theorem}\label{sec:syn-theory}
Suppose that there are two independent generalized open oscillators coupled to $m$ independent
bosonic fields, with $m$ output channels: an $n_1$ degrees of freedom oscillator $G_1=(S_1,L_1,H_1)$
with canonical operators $x_1=(q_{1,1},p_{1,1},\ldots,q_{1,n_1},p_{1,n_1})\trp$, Hamiltonian operator
$H_1=\half x^T R_1 x_1$, coupling operator $L_1=K_1 x_1$, and scattering matrix $S_1$, and, similarly,
an $n_2$ degrees of freedom oscillator $G_2=(S_2,L_2,H_2)$ with canonical operators
$x_2=(q_{2,1},p_{2,1},\ldots,q_{2,n_2},p_{2,n_2})\trp$, Hamiltonian operator
$H_2=\half x_2^T R_2 x_2$, coupling operator $L_2=K_2 x_2$, and unitary scattering matrix
$S_2$.

Consider now a reducible quantum network $\Ncal_{12}$ constructed from $G_1$ and $G_2$ as
$\Ncal_{12}=\{\{G_1,G_2\},H^d_{12},G_2 \triangleleft G_1\}$, as shown in Figure \ref{fig:G1G2network},
where $H^d_{12}$ is a direct interaction Hamiltonian term between $G_1$ and $G_2$ given by
\begin{eqnarray*}
H^d_{12}&=&\half x_1^T R_{12} x_2 + \half x_2\trp R_{12}\trp
x_1-\frac{1}{2i}(L_2^{\dagger}S_2
L_1-L_1^{\dagger}S_2^{\dagger}L_2);\, R_{12} \in \Rbb^{2 \times 2}\\
&=& x_2\trp R_{12}\trp x_1-\frac{1}{2i}(L_2^{\dagger}S_2
L_1-L_2\trp
S_2^{\#}L_1^{\#})\\
&=&x_2\trp
\left( R_{12}\trp-\frac{1}{2i}(K_2^{\dagger}S_2K_1-K_2\trp
S_2^{\#}K_1^{\#})\right)x_1,
\end{eqnarray*}
where we recall that $A^{\#}$ denotes the elementwise adjoint of a matrix of operators $A$ and the
second equality holds, since elements of $L_1$ commute with those $L_2$. Also note that the matrix
$\frac{1}{2i}(K_2^{\dagger}S_2K_1-K_2\trp S_2^{\#}K_1^{\#})$ is real.
\begin{figure}[b!]
\centerline{\includegraphics[width=12pc]{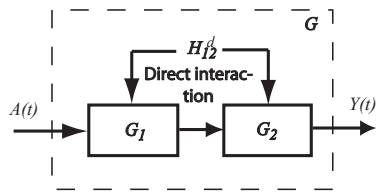}}
\caption{Cascade connection of $G_1$ and $G_2$ with indirect interaction $H_{12}^d$.\label{fig:G1G2network}}
\end{figure}
Some straightforward calculations (see \cite{GJ07} for details) then show that we may write
\begin{eqnarray*}
\Ncal_{12}&=&(S_2 S_1,S_2 L_1+L_2,H_1+H_2+H_{12}^f+H_{12}^d),
\end{eqnarray*}
where $H_{12}^f=\frac{1}{2i}(L_2^{\dagger}S_2 L_1-L_1^{\dagger}S_2^{\dagger}L_2)$. Now let us look
closely at the Hamiltonian term of $\Ncal_{12}$. Note that after plugging in the definition of
$H_1$, $H_2$, $H^d_{12}$, and $H_{12}^f$, we may write
\begin{eqnarray*}
H_1+H_2+H_{12}^f+H_{12}^d &=& \half [\begin{array}{cc} x_1\trp & x_2\trp
\end{array}]\left[
\begin{array}{cc} R_1 & R_{12} \\ R_{12}\trp & R_2 \end{array}\right]
\left[\begin{array}{c} x_1 \\ x_2 \end{array}\right].
\end{eqnarray*}
Letting $x=(x_1\trp,x_2\trp)\trp$, $S_{2\leftarrow 1}=S_2 S_1$, and defining
\begin{eqnarray}
R&=&\left[\begin{array}{cc} R_1 & R_{12} \\ R_{12}\trp & R_2\end{array}\right],\label{eq5.1}\\
K&=&[\begin{array}{cc} S_2 K_1 & K_2 \end{array}],\label{eq5.2}
\end{eqnarray}
we see that
\begin{equation}
\Ncal_{12}= \left(S_{2\leftarrow 1},K  x, \frac{1}{2}x \trp R x\right).
\label{eq5.3}
\end{equation}
Therefore, $\Ncal_{12}=(S_{2\leftarrow 1},L_{2\leftarrow 1},H_{2\leftarrow 1})$, with
$S_{2\leftarrow 1}=S_2S_1,\,L_{2\leftarrow 1}=K x,\,\hbox{and}\,H_{2\leftarrow 1}=\half x\trp R x$.
In other words, a reducible network formed by a bilinear direct interaction and cascade connection of
two generalized open oscillators having the same number of input and output fields results in another
generalized open oscillator with a degrees of freedom which is the sum of the degrees of freedom of
the two constituent oscillators and having the same number of inputs and outputs.

By repeated application of the above construction, we can prove the following synthesis theorem.

\begin{theorem}\label{th:synthesis}
Let $G$ be an $n$ degrees of freedom generalized open oscillator with Hamiltonian matrix
$R \in \Rbb^{2n \times 2n}$, coupling matrix $K \in \Cbb^{m \times 2n}$, and unitary scattering
matrix $S \in \Cbb^{m \times m}$. Let $R$ be written in terms of blocks of $2 \times 2$ matrices
as $R=[R_{jk}]_{j,k=1,\ldots,n}$, where the $R_{jk}$'s are real $2 \times 2$ matrices satisfying
$R_{kj}=R_{jk}\trp$ for all $j,k$, and let $K$ be written as
\[
K=[\begin{array}{cccc} K_1 & K_2 & \ldots & K_n
\end{array}],
\]
where, for each $j$, $K_j \in \Cbb^{m \times 2}$. For $j=1,\ldots,n$, let
$G_j=(S_j,\tilde K_j x_j, \half x_j^T R_{jj} x_j)$ be independent one degree of freedom generalized
open oscillators with canonical operators $x_j=(q_j,p_j)\trp$, $m$ output fields, Hamiltonian matrix
$R_{jj}$, coupling matrix $\tilde K_j$, and scattering matrix $S_j$. Also, define
$S_{k \twoheadleftarrow j}$ for $j\leq k+1$ as $S_{k \twoheadleftarrow j}=\prod_{l=j}^{k}
S_l=S_k \cdots S_{j+1}S_j$ for $j<k$, $S_{k \twoheadleftarrow k}=S_k$, and $S_{k
\twoheadleftarrow k+1}=I_{m \times m}$, and let $H^d$ be a direct interaction Hamiltonian given
by
\begin{eqnarray}
H^d&=&\sum_{j=1}^{n-1}\sum_{k=j+1}^{n}
x_k\trp\left(\vphantom{\frac{1}{2i}(\tilde K_k^{\dagger}S_{k
\twoheadleftarrow j+1}\tilde K_j-\tilde K_k\trp S_{k
\twoheadleftarrow j+1}^{\#}\tilde K_j^{\#})} R_{jk}\trp-
\frac{1}{2i}(\tilde K_k^{\dagger}S_{k \twoheadleftarrow j+1}\tilde
K_j-\tilde K_k\trp S_{k \twoheadleftarrow j+1}^{\#}\tilde
K_j^{\#})\right)x_j. \label{eq:multi-direct}
\end{eqnarray}
If $S_1,\ldots,S_n$ satisfies $S_{n} S_{n-1} \cdots S_1=S$ and $\tilde K_k$ satisfies
$\tilde K_k = S_{n \twoheadleftarrow k+1}^{\dagger}K_k$ for $k=1,\ldots,n$, then the reducible network
of harmonic oscillators $\Ncal$ given by $\Ncal=\{\{G_1,\ldots,G_n\},H^d,\{G_2 \triangleleft G_{1},G_3
\triangleleft G_2,\ldots,G_{n} \triangleleft G_{n-1}\}\}$ is equivalent to $G$. That is, $G$ can be
synthesized via a series connection $G_n \triangleleft \ldots \triangleleft G_2 \triangleleft G_1$ of
$n$ one degree of freedom generalized open oscillators, along with a suitable bilinear direct interaction
Hamiltonian involving the canonical operators of these oscillators. In particular, if $S=I_{m \times m}$
(no scattering), then $S_k$ can be chosen to be $S_k=I_{m \times m}$ and $\tilde K_k$ can be chosen to
be $\tilde K_k=K_k$ for $k=1,\ldots,n$.
\end{theorem}

\begin{proof}
Let $H_j=\half x_j\trp R_{jj}x_j$, $L_j=\tilde K_j x_j$ and
$$
H_{k}^f=\sum_{j=2}^{k}\left(L_{j}^{\dagger}\sum_{l=1}^{j-1}
S_{j \twoheadleftarrow l+1}L_l-\sum_{l=1}^{j-1} L_l^{\dagger}S_{j
\twoheadleftarrow l+1}^{\dagger} L_j\right),\quad k\geq 2.
$$

Let us begin with the series connection $G_{12}=G_2 \triangleleft G_1$. By analogous calculations as
given above for the two oscillator case, it is given by
$$
G_{12}=(S_2 S_1,S_2 L_1+L_2,H_1+H_2+H_{2}^f).
$$
Repeating this calculation recursively for $G_{123}=G_3 \triangleleft G_{12}$,
$G_{1234}=G_4 \triangleleft G_{123}, \ldots, \break G_{12\ldots n}=G_n \triangleleft G_{12\ldots (n-1)}$,
we obtain at the end that
$$
G_{12\ldots n }=\left(S_{n \twoheadleftarrow 1},\sum_{k=1}^{n} S_{n
\twoheadleftarrow k+1}L_k,\sum_{k=1}^{n}H_k+H_{n}^f\right).
$$
Noting that $H_n^f$ may be rewritten as
\begin{eqnarray*}
H_n^f&=&\frac{1}{2i}\sum_{j=1}^{n-1}\sum_{k=j+1}^{n}(L_k^{\dagger}S_{k
\twoheadleftarrow j+1}L_j-L_j^{\dagger}S_{k
\twoheadleftarrow j+1}^{\dagger}L_k)\\
&=&\frac{1}{2i}\sum_{j=1}^{n-1}\sum_{k=j+1}^{n} (L_k^{\dagger}S_{k
\twoheadleftarrow j+1}L_j-L_k\trp S_{k
\twoheadleftarrow j+1}^{\#}L_j^{\#})\\
&=&\frac{1}{2i}\sum_{j=1}^{n-1}\sum_{k=j+1}^{n} x_k\trp(\tilde
K_k^{\dagger}S_{k \twoheadleftarrow j+1}\tilde K_j-\tilde K_k\trp
S_{k \twoheadleftarrow j+1}^{\#}\tilde K_j^{\#})x_j,
\end{eqnarray*}
where the second equality holds since $L_j$ commutes with $L_k$ whenever $j \neq k$, we find that
\begin{eqnarray*}
\sum_{k=1}^nH_k+H_n^f +H^d &=&\half \sum_{j=1}^{n}\sum_{k=1}^{n}x_j
R_{jk}x_k =\half x\trp R x,\quad x=(x_1\trp,x_2\trp,\ldots,x_n\trp)\trp.
\end{eqnarray*}
Therefore, if $S_1,\ldots,S_n$ and $\tilde K_1,\ldots,\tilde K_n$ satisfy the conditions stated in the
theorem, we find that $\Ncal=\{\{G_1,\ldots,G_n\},H_d,\{G_2 \triangleleft G_1,G_3 \triangleleft
G_2,\ldots,G_n \triangleleft G_{n-1}\}$ is given by
\begin{eqnarray*}
\Ncal &=& \left(S,K x,\half x\trp R x \right).
\end{eqnarray*}
That is, $\Ncal$ is a linear quantum stochastic system with Hamiltonian matrix $R$, coupling matrix
$K$, and scattering matrix $S$, and is therefore equivalent to $G$. This completes the proof of the
synthesis theorem.%
\qquad\end{proof}

Therefore, according to the theorem, synthesis of an arbitrary $n$ degrees of freedom linear dynamical
quantum stochastic system is in principle possible if the following two requirements can be met:
\begin{enumerate}
\item Arbitrary one degree of freedom open oscillators $G=(I,L,H)$ with $m$ input and output fields
can be synthesized. In particular, it follows from this that one degree of freedom generalized open
oscillators $G'=(S,L,H)$ can be synthesized as $G'=(I,L,H) \triangleleft (S,0,0)$.

\item The bilinear interaction Hamiltonian $H^d$ as given by (\ref{eq:multi-direct}) can be synthesized.
\end{enumerate}

One can observe certain parallels between the quantum synthesis described in the theorem with the
active state-space synthesis method of linear electrical network synthesis theory (e.g.,
\cite[Chapter 13]{AV73}). To begin  with, we may think of  each oscillator $G_j$ as a type of noisy
quantum integrator, as the  counterpart of a classical integrator (built from an operational amplifier,
resistors, and capacitors) in an electrical network, and in both settings synthesis can be achieved by
suitably cascading these basic integrating components. We may also view the direct interaction Hamiltonian
between two oscillators as acting like a type of  mutual ``state feedback'' between the oscillators,
much like the state feedback employed in electrical network synthesis. However, because of the quite
distinct nature of classical and quantum systems, of course the parallels should not be taken to be
``exact'' or ``precise'' in any way, the nature of these parallels are in spirit rather than detail.
Certainly, classical active synthesis theory cannot be applied directly to linear quantum stochastic
systems because of certain constraints that a noisy quantum integrator must satisfy that are not
required of its classical counterpart, and the classical theory is deterministic while in the quantum
theory,  quantum stochastic noise plays a crucial role, for instance, to preserve the canonical commutation
relations in open quantum systems. To highlight another significant difference between the two physical
systems, we note that losses in linear electrical systems may be  modeled by inserting resistors as
dissipative components of the system, while in linear quantum systems,  losses are modeled by lossy
couplings to quantum noises (heat baths).

\section{Systematic synthesis of linear quantum stochastic systems}\label{sec:system-synth}
This section details the construction of arbitrary one degree of freedom open quantum harmonic oscillators
and implementation of bilinear direct interactions among the canonical operators of these oscillators,
at least approximately, in the context of quantum optics, using various linear and nonlinear quantum
optical components.

We begin with a description of some key quantum optical components that will be required for the
synthesis. This is followed by a discussion of general synthesis of one degree of freedom open
oscillators and finally by a discussion of the implementation of bilinear direct interaction Hamiltonians
among different one degree of freedom open oscillators.

\subsection{Essential quantum optical components}\label{sec6.1}
\subsubsection{Optical cavities}\label{sec:opt-cav}
An optical cavity is a system of fully reflecting or partially transmitting mirrors in which a light
beam is trapped and is either bounced repeatedly from the mirrors to form a standing wave or circulates
inside the cavity (as in a ring cavity); see Figure \ref{fig:ring-cav}. If there are transmitting mirrors
present, then light can escape or leak out from the cavity, introducing losses to the cavity.

\begin{figure}[tbp!]
\centerline{\includegraphics[width=18pc]{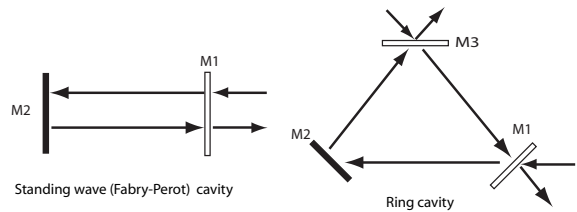}}
\caption{Two types of optical cavities: a standing wave or Fabry--Perot cavity (left) and a
(three mirror) ring cavity (right). Arrows indicate the direction of propagation of light in
the cavity. Black rectangles denote fully reflecting mirrors, while white rectangles denote partially
transmitting mirrors.\label{fig:ring-cav}}
\end{figure}

A cavity is mathematically modeled by a Hamiltonian $H_{cav}=\omega_{cav} a^*a$, where $\omega_{cav}$
is the resonance frequency of the cavity and $a=\frac{q+ip}{2}$ is the (non-self-adjoint) cavity
annihilation operator or cavity mode satisfying the commutation relation $[a,a^*]=1$. Here $q=a+a^*$
is the position operator of the cavity mode (also called the {\em amplitude quadrature} of the mode)
and $p=-ia+ia^*$ is the momentum operator of the cavity mode (also called the {\em phase quadrature}
of the mode). If there is a transmission mirror, say, M, then losses through this mirror are modeled as
having a vacuum bosonic noise field $A(t)$ incident at this mirror and interacting with the cavity
mode via the idealized Hamiltonian $H_{Int}$ given in (\ref{eq:ideal-int}) with $L=\sqrt{\kappa}a$,
where $\kappa$ is a positive constant called the mirror \emph{coupling coefficient}. When there are
several leaky mirrors, then the losses are modeled by a sum of such interaction Hamiltonians, one for
each mirror and with each mirror having its own distinct vacuum bosonic field. The total Hamiltonian
of the cavity is then just the sum of $H_{cav}$ and the interaction Hamiltonians. More generally, the
field incident at a transmitting mirror need not be a vacuum field, but can be other types of fields,
such as a coherent laser field. Nonetheless, the interaction of the cavity mode with such fields
via the mirror will still be governed by (\ref{eq:ideal-int}) with a coupling operator of the form
$L=\sqrt{\kappa}a$.

\subsubsection{Degenerate parametric amplifier}\label{sec:DPA}
In order to amplify a quadrature of the cavity mode, for example, to counter losses in that quadrature
caused by light escaping through a transmitting mirror, one can employ a $\chi^{(2)}$ nonlinear
optical crystal and a classical pump beam in the configuration of a degenerate parametric amplifier
(DPA), following the treatment in \cite[section 10.2]{GZ00}. The pump beam acts as a source of
additional quanta for amplification and, in the nonlinear crystal, an interaction takes place in which
photons of the pump beam are annihilated to create photons of the cavity mode. In an optical cavity,
such as a ring cavity shown in Figure \ref{fig:DPA-cav}, we place the crystal in one arm of the cavity
(for example, in the arm between mirrors M1 and M2) and shine the crystal with a strong coherent pump
beam of (angular) frequency $\omega_p$ given by $\omega_p=2\omega_r$, where $\omega_r$ is some reference
frequency. Here the mirrors at the end the arms should be chosen such that they do not reflect light
beams of frequency $\omega_p$. A schematic representation of a DPA (a nonlinear crystal with a classical
pump) is shown in Figure \ref{fig:DPA}.

\begin{figure}[tbp!]
\centerline{\includegraphics[width=9pc]{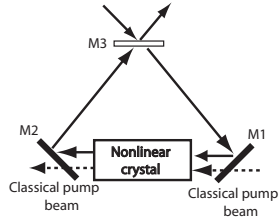}}
\caption{A DPA consisting of a classically pumped nonlinear crystal in
a three mirror ring cavity.\label{fig:DPA-cav}}
\end{figure}

\begin{figure}[tbp!]
\centerline{\includegraphics[width=5.3pc]{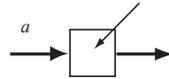}}
\caption{Schematic representation of a DPA. The white rectangle symbolizes the nonlinear crystal,
while the diagonal arrow into the rectangle denotes the pump beam.\label{fig:DPA}}
\end{figure}

\begin{remark}\label{remark6.1}\rm
In the remaining figures, black rectangles  will be used to denote mirrors which are fully reflecting
at the cavity frequency and fully transmitting at the pump frequency (whenever a pump beam is employed),
while white rectangles denote partially transmitting mirrors at the cavity frequency.
\end{remark}

Let $a=\frac{q+ip}{2}$ be the cavity mode, and let the cavity frequency $\omega_{cav}$ be detuned from
$\omega_r$ and given by $\omega_{cav}=\omega_r+\Delta$, where $\Delta$ is the frequency detuning. The
crystal facilitates an energy exchange interaction between the cavity mode and pump beam. By the
assumption that the pump beam is intense and  not depleted in this interaction, it may be assumed to
be classical, in which case the crystal-pump-cavity interaction can be modeled using the (time-varying)
Hamiltonian $H(t)=\omega_{cav}a^{*}a+\frac{i}{2} (\epsilon e^{-i\omega_p t}
(a^{*})^2-\epsilon^*e^{i\omega_p t}a^2)$ \cite[equation 10.2.1]{GZ00}, where $\epsilon$ is a complex number
representing the {\em effective} pump intensity. By transforming to a rotating frame with respect to
$\omega_r=\frac{\omega_p}{2}$ (i.e., by application of the transformation $a \mapsto a
e^{i\frac{\omega_p}{2}t}$; see \cite[section 10.2.1]{GZ00} for a derivation of the equations of motion
of the DPA in the rotating frame), $H$ can be reexpressed as
$$
H=\Delta a^* a+\frac{i}{2} (\epsilon (a^{*})^2-\epsilon^*a^2)
$$
and be written compactly as $H=\half x_0^T R x_0+c$ (recall $x_0=(q,p)^T$), where
\begin{eqnarray}
R=\half \left[\begin{array}{cc} \Delta+\dst\frac{i}{2}(\epsilon-\epsilon^*)
& \dst\frac{1}{2}(\epsilon+\epsilon^*) \\[6pt]
\dst\frac{1}{2}(\epsilon+\epsilon^*) &
\dst \Delta-\frac{i}{2}(\epsilon-\epsilon^*)
\end{array}\right] \label{eq:R-DPA}
\end{eqnarray}
and $c$ is a real number. Since $c$ merely contributes a phase factor that has no effect on the overall
dynamics of the system operators, it plays no essential role and can simply be ignored (cf.\ section
\ref{sec:models}). Note that transformation to a rotating frame effects the following: If $a(t)$ is the
evolution of $a$ under the original time-varying Hamiltonian $H(t)=\omega_{cav}a^{*}a+\frac{i}{2}
(\epsilon e^{-i\omega_p t} (a^{*})^2-\epsilon^*e^{i\omega_p t}a^2)$ and we define
$\tilde a(t) = a(t)e^{i \omega_r t}$ (i.e., $\tilde a(t)$ is $a(t)$ in a frame rotating at frequency
$\omega_r$), then $\tilde a(t)$ coincides with the time evolution of $a$ under the time-independent
Hamiltonian $\tilde H=\half x_0^T R x_0$. In other words, in this rotating frame, the DPA can be viewed
as a harmonic oscillator with quadratic Hamiltonian $\half x_0^T R x_0$.

\subsubsection{Two-mode squeezing}\label{sec:two-mode}
If two cavities are positioned in such a way that the beams circulating in them intersect one another,
then these beams will merely pass through each other without interacting. One way of making the beams
interact is to have their paths intersect inside a $\chi^{(2)}$ nonlinear optical crystal. Typically,
to facilitate such an interaction, one or two auxiliary pump beams are also employed as a source of
quanta/energy. For instance, in a $\chi^{(2)}$ optical crystal in which the modes of two cavities
interact with an undepleted classical pump beam as depicted in Figure \ref{fig:two-mode}, the
interaction can be modeled by the Hamiltonian
$$
H(t)=\frac{i}{2}(\epsilon e^{-i\omega_p t}a_1^*a_2^*-\epsilon^*
e^{i\omega_p t} a_1 a_2),
$$
where $\epsilon$ is a complex number representing the effective intensity of the pump beam and $\omega_p$
is the pump frequency. Transforming to a rotating frame at half the pump frequency by applying the rotating
frame transformation $a_1 \mapsto a_1e^{i\frac{\omega_p}{2}t}$ and $a_2 \mapsto a_2 e^{i\frac{\omega_p}{2}t}$,
$H(t)$ can be expressed in this new frame in the time-invariant form
$H=\frac{i}{2}(\epsilon a_1^*a_2^*-\epsilon^*a_1 a_2)$. This type of Hamiltonian is called a {\em two-mode
squeezing Hamiltonian}, as it simultaneously affects squeezing in one quadrature of (possibly rotated versions
of) $a_1$ and $a_2$ and will play an important role later on in the paper. A two-mode squeezer is
schematically represented by the symbol shown in Figure \ref{fig:schem-tms}.

\begin{figure}[tbp!]
\centerline{\includegraphics[width=9pc]{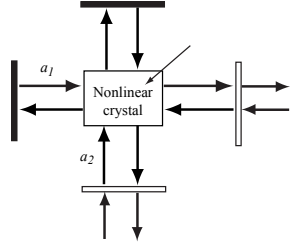}}
\caption{Two cavity modes interacting in a single classically pumped nonlinear crystal. The diagonal
arrow into the crystal denotes the pump beam.\label{fig:two-mode}}
\vspace{8pt}
\end{figure}

\begin{figure}[tbp!]
\centerline{\includegraphics[width=5.2pc]{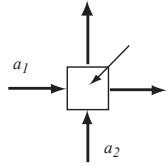}}
\caption{Schematic representation of a two-mode squeezer.\label{fig:schem-tms}}
\end{figure}

\begin{remark}\label{rem:ref-frame}\rm
It will be implicitly assumed in this paper that the equations for the dynamics of generalized open
operators are given with respect to a common rotating frame of frequency $\omega_r$, including the
transformation of all bosonic noises $A_i(t)$ according to $A_i(t) \mapsto A_i(t)e^{i\omega_r t}$,
and that classical pumps employed are all of frequency $\omega_p=2\omega_r$. This is a natural setting
in quantum optics where a rotating frame is essential for obtaining linear time invariant QSDE models
for active devices that require an external source of quanta. In a control setting, this means both
the  quantum plant and the controller equations have been expressed in the same rotating frame.
\end{remark}

\subsection{Static linear optical devices and networks}\label{sec:static}
Static linear optical devices implement static linear transformations (meaning that the transformation
can be represented by a complex square matrix) of a set of independent incoming single mode fields,
such as the field in a cavity, $a=(a_1,a_2,\ldots a_m)^T$ to an equal number $a'=(a_1',a_2',\ldots a_m')^T$
of independent outgoing fields. The incoming fields satisfy the commutation relations $[a_j,a_k]=0$
and $[a_j,a_k^*]= \delta_{jk}$. The incoming fields may also be vacuum bosonic fields
$A(t)=(A_1(t),A_2(t),\ldots,A_m(t))^T$ with outgoing bosonic fields (that need no longer be in the
vacuum state) $A'(t)=(A_1'(t),A_2'(t),\ldots,A_m'(t))^T$. In the latter,  the commutation relations are
$[dA_j(t),dA_k(t)]=0$ and $[dA_j(t),dA_k(t)^*]=\delta_{jk} dt$. However, to avoid cumbersome and
unnecessary repetitions, in the following we shall only discuss the operation of a static linear
optical device in the context of single mode fields. The operation is completely analogous for bosonic
incoming and outgoing fields and  requires only making substitutions such as $a\rightarrow A(t)$,
$a'\rightarrow A'(t)$, $[a_j,a_k]=0\rightarrow [dA_j(t),dA_k(t)]=0$, and
$[a_j,a_k^*]=\delta_{jk} \rightarrow [dA_j(t),dA_k(t)^*]=\delta_{jk}dt$, etc.

The operation of a static linear optical device can mathematically be expressed as
\begin{eqnarray*}
\left[\begin{array}{l} a' \\ a'^{\#} \end{array}\right]&=&Q
\left[\begin{array}{l} a \\ a^{\#}
\end{array}\right];\; Q=\left[\begin{array}{cc} Q_1 & Q_2 \\
Q_2^{\#} &  Q_1^{\#} \end{array}\right],
\end{eqnarray*}
where $Q_1,Q_2 \in \Cbb^{m \times m}$ and $S$ is a quasi-unitary matrix \cite[section 3.1]{Leon03}
satisfying
$$
Q\left[\begin{array}{cc}I & 0 \\ 0 & -I
\end{array}\right]Q^{\dag}= \left[\begin{array}{cc} I & 0 \\ 0 &
-I
\end{array}\right].
$$
A consequence of the quasi-unitarity of $Q$ is that it preserves the commutation relations among the
fields, that is, to say that the output fields $a'$ satisfy the same commutation relations as $a$.
Another important property of a quasi-unitary matrix is that it has an inverse $Q^{-1}$ given by
$Q^{-1}=GQ^{\dag}G$, where $G=\mbox{\scriptsize$\left[\begin{array}{@{}cc@{}}I & 0 \\ 0 & -I \end{array}\right]$}$, and this
inverse is again quasi-unitary, i.e., the set of quasi-unitary matrices of the same dimension form
a group.

In the case where the submatrix $Q_2$ of $Q$ is $Q_2=0$, the device does not mix creation and
annihilation operators of the fields, and it necessarily follows that $Q_1$ is a complex {\em unitary}
matrix. Such devices are said to be {\em static passive} linear optical devices because they do not
require any external source of quanta for their operation. It is well known that any passive network
can be constructed using only beam splitters and mirrors (e.g., see references 2--4 in \cite{LN04}).
In all other cases, the devices are {\em static active}. Specific passive and static devices that will
be utilized in this paper will be discussed in the following.

\subsubsection{Phase shifter}\label{sec:phase-shifter}
A phase shifter is a device that produces an outgoing field that is a phase shifted version of the
incoming field. That is, if there is one input field $a$, then the output field is $a'=e^{i\theta}a$
for some real number $\theta$, called the {\em phase shift}; a phase shifter is schematically represented
by the symbol shown in Figure \ref{fig:phase-shifter}. By definition, a phase shifter is a static
passive device. The transformation matrix $Q_{PS}$ of a phase shifter with a single input field is
given by
$$
Q_{PS}=\left[\begin{array}{cc} e^{i\theta} & 0 \\ 0 &
e^{-i\theta}\end{array} \right].
$$

\begin{figure}[tb]
\centerline{\includegraphics[width=7pc]{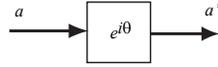}}
\caption{Phase shifter with a phase shift of $\theta$ radians.\label{fig:phase-shifter}}
\end{figure}

\subsubsection{Beam splitter}\label{sec:BS}
A beam splitter is a static and passive device that forms a linear combination  of two input fields
$a_1$ and $a_2$ to produce two output fields $a_1'$ and $a_2'$ such that energy is conserved:
$a_1^*a_1+a_2^*a_2=(a_1')^* a_1' + (a_2')^*a_2'$. The transformation affected by a beam splitter can
be written as
$$
Q_{BS}=\left[\begin{array}{cc} B & 0 \\ 0 & B^{\#}
\end{array}\right],
$$
where $B$ is a unitary matrix given by
$$
B=e^{i\Xi/2}\left[\begin{array}{cc} e^{i\Psi/2} & 0
\\ 0 & e^{-i\Psi/2} \end{array}\right]\left[\begin{array}{cc} \cos(\Theta)/2 & \sin(\Theta)/2 \\
-\sin(\Theta)/2
& \cos(\Theta)/2\end{array}\right]\left[\begin{array}{cc}
e^{i\Phi/2} & 0
\\ 0 & e^{-i\Phi/2} \end{array}\right].
$$
Here $\Xi,\Theta,\Phi,\Psi$ are real numbers. $\Theta$ is called the {\em mixing angle} of the beam
splitter, and it is the most important parameter. $\Phi$ and $\Psi$ introduce a phase difference in the
two incoming and outgoing modes, respectively, while $\Xi$ introduces an overall phase shift in both
modes.

A particularly useful result on the operation of a beam splitter with $\Xi=\Psi=\Phi=0$ is that it can
be modeled by an effective Hamiltonian $H_{BS}^0$ given by $H_{BS}^0=i\Theta(a_1^*a_2-a_1 a_2^*)$ (see
\cite[section 4.1]{Leon03} for details). This means that in this case we have the representation
$$
Q_{BS}\left[\begin{array}{l} a \\ a^{\#}
\end{array}\right]=\exp(iH_{BS}^0) \left[\begin{array}{l} a \\ a^{\#}
\end{array}\right]\exp(-iH_{BS}^0),
$$
where $a=(a_1,a_2)^T$. More generally, it follows from this, by considering phase shifted inputs
$a_1 \rightarrow a_1 e^{i\frac{\theta+\Phi}{2}}$ and $a_2 \rightarrow a_2 e^{i\frac{\theta-\Phi}{2}}$
($\theta$ being an arbitrary real number), that a beam splitter with $\Xi=0$ and $\Psi=-\Phi$ will
have the effective Hamiltonian $H_{BS}=i\Theta (e^{-i\Phi}a_1^*a_2-e^{i\Phi}a_1 a_2^*)=\alpha a_1^*a_2
+\alpha^* a_1 a_2$, with $\alpha=i\Theta e^{-i\Phi} $. This is the most general type of beam splitter
that will be employed in the realization theory of this paper. A beam splitter with a Hamiltonian of
the form $H_{BS}$ is represented schematically using the symbol in Figure \ref{fig:bs}.

\begin{figure}[tbp!]
\centerline{\includegraphics[width=4.5pc]{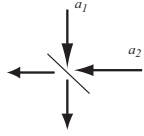}}
\caption{Schematic representation of a beam splitter.\label{fig:bs}}
\end{figure}

\subsubsection{Squeezer}\label{sec:squeezer}
Let there be a single input mode $a$. Write $a$ as $a=q'+ip'$, where  $q'=(a+a^*)/2$ is the {\em real}
or {\em amplitude} quadrature of $a$ and $p'=(a-a^*)/2i$ is the {\em imaginary} or {\em phase} quadrature
of $a$. {\em Squeezing} of a field is an operation in which the variance of one quadrature, either
$q'$ or $p'$, is squeezed or attenuated (it becomes less noisy) at the expense of increasing the
variance of the other quadrature (it becomes noisier). A device that performs squeezing of a field
is called a {\em squeezer}. An ideal squeezer affects the transformation $Q_{squeezer}$ given by
$$
Q_{squeezer}=\left[\begin{array}{cc} \cosh(s) & e^{i\theta}
\sinh(s)
\\ e^{-i\theta}\sinh(s) & \cosh(s) \end{array}\right],
$$
where $s$ and  $\theta$ are real parameters. We shall refer to $s$ as the squeezing parameter and
$\theta$ as the phase angle. For $s<0$, the squeezer squeezes the amplitude quadrature of
$e^{-i\frac{\theta}{2}}a$ (a phase shifted version of $a$) while if $s>0$, it squeezes the phase
quadrature and then shifts the phase of the squeezed field by $\frac{\theta}{2}$. A squeezer with
parameters $s,\theta$ is schematically represented by the symbol shown in Figure \ref{fig:squeezer}.

\begin{figure}[tbp!]
\centerline{\includegraphics[width=6.5pc]{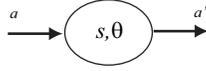}}
\caption{Schematic representation of a squeezer.\label{fig:squeezer}}
\end{figure}

A squeezer can be implemented, for instance, by using a combination of a parametric amplifier and a
beam splitter for single mode fields \cite[section 6.1]{Leon03}  or as a DPA with a transmitting mirror
for bosonic fields \cite[section 7.2.9]{GZ00}. It is easy to see that $Q_{squeezer}^{-1}$ is given by
$$
Q_{squeezer}^{-1}=\left[\begin{array}{cc} \cosh(-s) & e^{i \theta}
\sinh(-s)
\\ e^{-i \theta }\sinh(-s) & \cosh(-s) \end{array}\right]=\left[\begin{array}{cc} \cosh(s) & -e^{i \theta}
\sinh(s)
\\ -e^{-i \theta }\sinh(s) & \cosh(s) \end{array}\right].
$$

\subsubsection{Static optical linear networks}\label{sec6.2.4}
It is known that an arbitrary static linear optical network can be decomposed as a cascade of simpler
networks. In particular, any quasi-unitary matrix $Q$ can be constructively decomposed as \cite{LN04}:
\begin{eqnarray*}
Q&=&\exp\left[\begin{array}{cc} A_1 & 0 \\ 0 & A_1^{\#}
\end{array}
\right] \exp\left[ \begin{array}{cc} 0 & D \\ D & 0 \end{array}  \right] \exp\left[ \begin{array}{cc} A_3 & 0 \\
0 & A_3^{\#} \end{array}\right]\\
&=&\left[\begin{array}{cc} \exp A_1  & 0 \\ 0 & \exp A_1^{\#}
\end{array}
\right] \left[\begin{array}{cc} \cosh D & \sinh D  \\ \sinh D & \cosh D \end{array}  \right]
\left[ \begin{array}{cc} \exp A_3 & 0 \\
0 & \exp A_3^{\#} \end{array}\right],
\end{eqnarray*}
where $A_1$ and $A_3$ are skew symmetric complex matrices and $D$ is a real diagonal matrix.  The first
and third matrix exponential represent  passive static networks that can be implemented by beam splitters
and mirrors, while the second exponential represents an independent collection of squeezers (with trivial
phase angles) each acting on a distinct field.

In summary, in any static linear optical network the incident fields can be thought of as going through
a sequence of three operations: they are initially mixed by a passive network, then they undergo squeezing,
and finally they are subjected to another passive transformation. In the special case where the entire
network is passive, the squeezing parameters (i.e., elements of the $D$ matrix) are zero.

For example, a squeezer with arbitrary phase angle $\theta$ can be constructed by sandwiching a squeezer
with phase angle $0$ between a $-\theta/2$ phase shifter at its input and a $\theta/2$ phase shifter at
its output, respectively. This is shown in Figure \ref{fig:squeezer-imp}.

\begin{figure}[tbp!]
\centerline{\includegraphics[width=12.5pc]{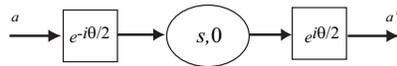}}
\caption{Implementation of a squeezer with arbitrary phase angle employing a squeezer with a zero phase
angle and two phase shifters.\label{fig:squeezer-imp}}
\end{figure}

\subsection{Synthesis of one degree of freedom open oscillators}\label{sec:one-degree-syn}
One degree of freedom open oscillators are completely described by a real symmetric Hamiltonian matrix
$R=R^T \in \Rbb^{2 \times 2}$ and complex coupling matrix $K \in \Cbb^{m \times 2}$. Thus one needs to
be able to implement both $R$ and $K$. Here we shall propose the realization of one degree of freedom
open quantum harmonic oscillators based around a ring cavity structure, such as shown in
Figure \ref{fig:ring-cav}, using fully reflecting and partially reflecting mirrors and nonlinear
optical elements appropriately placed between the mirrors.

The matrix $R$ determines the quadratic Hamiltonian $H=\half x^T R x$ and in a one-dimensional setup
such a quadratic Hamiltonian can be realized with a DPA as discussed
in section \ref{sec:DPA}. From (\ref{eq:R-DPA}), it is easily inspected that any real symmetric matrix
$R$ can be realized by suitably choosing the complex effective pump intensity parameter $\epsilon$ and the
cavity detuning parameter $\Delta$ of the DPA. In fact, for any particular $R$, the choice of parameters
is {\em unique}. For example, to realize
$$
R=\left[\begin{array}{cc} 1 & -2 \\
-2 & 0.5\end{array} \right],
$$
one solves the set of equations
$$
\Delta-\Im\{\epsilon\}=2,\quad  \Re\{\epsilon\}=-4, \quad \mbox{and}\quad
\Delta+\Im\{\epsilon\}=1
$$
for $\Delta,\Im\{\epsilon\},\Re\{\epsilon\}$ to yield the unique solution $\Delta=3/2$ and
$\epsilon=-4-i/2$.

Now, we turn to consider realization of the coupling operator $L=Kx_0$. Let us write
$K=[\begin{array}{ccc} K_1^T & \ldots & K_m^T\end{array}]^T$, where $K_l \in \Cbb^{1 \times 2}$ for
each $l=1,\ldots,m$. Each $K_l$ represents the coupling of the oscillator to the bosonic field $A_l$,
and so it suffices to study how to implement a single linear coupling to just one field. To this end,
suppose now that there is only one field $A(t)$ coupled to the oscillator via a linear coupling operator
$L=Kx_0$ for some $K \in \Cbb^{1 \times 2}$. It will be more convenient to express $L=\alpha q + \beta p$
in terms of the oscillator annihilation operator $a$ and creation operator $a^*$ defined by $a=(q+ip)/2$
and $a^*=(q-ip)/2$. Therefore, we write $L=\tilde \alpha a+\tilde \beta  a^*$, with
$\tilde \alpha=\frac{\alpha - i\beta}{2}$ and $\tilde\beta=\frac{\alpha + i\beta}{2}$. Consider the
physical scheme shown in Figure \ref{fig:diss-coup-1}, partly inspired by a scheme proposed by Wiseman
and Milburn for quantum nondemolition measurement of the position operator, treated at the level
of master equations \cite{WM93} (whereas here we consider unitary models and QSDEs).
\begin{figure}[tbp!]
\centerline{\includegraphics[width=18pc]{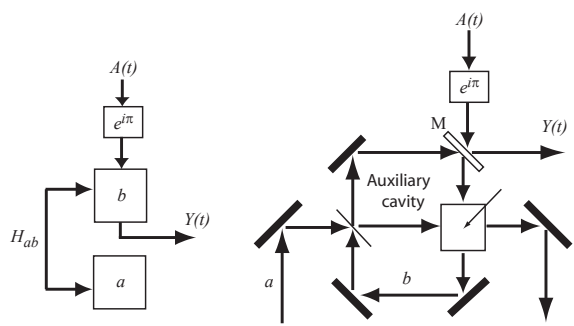}}
\caption{Scheme for (approximate) implementation of a coupling $L=\tilde \alpha a + \tilde \beta a^*$
to cavity mode $a$ using an auxiliary cavity $b$ (whose dynamics is adiabatically eliminated), a
two-mode squeezer, and a beam splitter with the appropriate parameters. The left figure is a block
diagram showing the fast mode $b$ interacting with the slow mode $a$ via the direct interaction
Hamiltonian $H_{ab}$, implemented by the two-mode squeezer and the beam splitter, and also interacting
with a $180^{\circ}$ phase shifted input field $A(t)$ to produce the output field $Y(t)$. The right
figure details the physical implementation of the block diagram.\label{fig:diss-coup-1}}
\end{figure}
In this scheme, additional mirrors are used to implement an auxiliary cavity mode $b$ of the same
frequency as the reference frequency $\omega_r$ (cf.\ Remark \ref{rem:ref-frame}). The auxiliary cavity
$b$ interacts with $a$ via a cascade of a two-mode squeezer and a beam splitter. The combination of the
nonlinear crystal and beam splitter implements an overall interaction Hamiltonian $H_{a b}$, in a
rotating frame at frequency $\omega_r$ (equal to half the pump frequency of the two-mode squeezer), of
the form
\begin{eqnarray}
H_{a b}=\frac{i}{2}(\epsilon_1 a^*b^*-\epsilon_1^* a b) +
\frac{i}{2} (\epsilon_2 a^* b - \epsilon_2^* a b^*),
\label{eq:tms-bs}
\end{eqnarray}
where $\epsilon_1$ is the effective pump intensity of the two-mode squeezer and $\epsilon_2$ is given
by $\epsilon_2=2\Theta e^{-i\Phi}$, where $\Theta$ is the mixing angle of the beam splitter and $\Phi$
is the relative phase introduced between the input fields by the beam splitter. Assuming that the coupling
coefficient $\gamma_2$ of the partially transmitting mirror $M$ on $b$ is such that $b$ is heavily damped
compared to $a$, $b$ will have much faster dynamics than $a$ and thus allows one to adiabatically eliminate
$b$ to obtain a reduced dynamics for $a$ only. A rigorous foundation for such adiabatic elimination or
singular perturbation procedure has recently been developed in \cite{BvHS07}. Based on this theory, the
adiabatic elimination results developed in Appendix \ref{sec:App-A} show that, after the  elimination of $b$,
the resulting coupling operator to $a$ will be given by
$$
L=\frac{1}{\sqrt{\gamma_2}} (-\epsilon_2^* a+\epsilon_1 a^*).
$$
Therefore, it becomes clear that by choosing the parameters $\epsilon_1,\epsilon_2,\gamma_2$ with
$\gamma_2$ large and such that
\begin{equation}
\tilde \alpha=-\frac{\epsilon_2^*}{\sqrt{\gamma_2}}\;\hbox{and}\;
\tilde \beta=\frac{\epsilon_1}{\sqrt{\gamma_2}},
\label{eq:coup-cond}
\end{equation}
it is possible to approximately implement any coefficients $\tilde \alpha$ and $\tilde \beta$ in a
linear coupling operator $L=\tilde \alpha a+\tilde \beta a^*$. Note that a $\pi$ radian phase shifter
in front of $A$ in Figure \ref{fig:diss-coup-1} is required to compensate for the scattering term in
the unitary model that is obtained after adiabatic elimination (cf.\ Appendix \ref{sec:App-A}).

Moreover, for the special case where $\tilde \alpha,\tilde \beta$ satisfy $\tilde \alpha$ is real and
$\tilde \alpha>|\tilde \beta|\geq 0$ we also propose an alternative implementation of the linear coupling
based on preprocessing and postprocessing with squeezed bosonic fields (see Appendix \ref{sec:App-B} for details).
To this end, let $\gamma= \tilde\alpha^2-|\tilde\beta|^2>0$, and consider the interaction Hamiltonian
\begin{eqnarray*}
H_{Int}(t) &=& i(L\eta(t)^*-L^*\eta(t)) \\
&=&i((\tilde\alpha a+\tilde\beta a^*)\eta(t)^*-(\tilde\alpha
a^*+\tilde\beta^* a)\eta(t)).
\end{eqnarray*}
Let us rewrite this Hamiltonian as follows:
\begin{eqnarray*}
H_{Int}(t) &=& i(a(\tilde\alpha
\eta(t)^*-\tilde\beta^*\eta(t))- a^*(\tilde\alpha \eta(t)-\tilde\beta \eta(t)^*))\\
&=&i\sqrt{\gamma}(a\eta'(t)^*-a^*\eta'(t)),
\end{eqnarray*}
where $\eta'(t)=\frac{1}{\sqrt{\gamma}}(\tilde\alpha \eta(t)-\tilde\beta \eta(t)^*)$. Letting
$Z(t)=\int_{0}^{t}\eta'(s)ds$, we have that $Z(t)=\frac{1}{\sqrt{\gamma}}(\tilde \alpha A(t) -\tilde \beta
A(t)^*)$, and
$$
\left[\begin{array}{c} Z(t)\\  Z(t)^*
\end{array}\right]=Q\left[\begin{array}{c}  A(t)\\  A(t)^*
\end{array}\right],\quad Q=\left[\begin{array}{cc} \frac{\tilde \alpha}{\sqrt{\gamma}}  & -\frac{\tilde
\beta}{\sqrt{\gamma}} \\ -\frac{\tilde \beta^*}{\sqrt{\gamma}} &
\frac{\tilde\alpha}{\sqrt{\gamma}}
\end{array}\right].
$$
The main idea is that instead of considering an oscillator interacting with $A(t)$, we consider the same
oscillator interacting with the new field $Z(t)$ via the interaction Hamiltonian
$H_{Int}(t)=i\sqrt{\gamma}(a \eta'(t)^*-a^* \eta'(t))$. Since $ \tilde \alpha ^2-|\tilde \beta|^2=\gamma>0$,
we see that $(\alpha/\sqrt{\gamma})^2-|\beta/\sqrt{\gamma}|^2=1$, from which it follows that $Q$ is a
quasi-unitary linear transformation (cf.\ section \ref{sec:static}) that preserves the field commutation
relations. In fact, $Z(t)$ by definition is a squeezed version of $A(t)$ that can be obtained from $A(t)$
by passing the latter through a squeezer with the appropriate parameters (cf.\ section \ref{sec:squeezer});
in this case the squeezer would have the parameters $s=-{\rm arccosh}(\tilde \alpha/\sqrt{\gamma})$ and
$\theta=\arg{\tilde \beta}$. $Z(t)$ satisfies $[dZ(t),dZ(t)^*]=dt$ and the Ito rules for a squeezed field
that can be generated from the vacuum (the theoretical basis for these manipulations are discussed in
Appendix \ref{sec:App-B}) are
$$
\left[\begin{array}{c} dZ(t)\\ dZ(t)^*
\end{array}\right]\left[\begin{array}{cc} dZ(t)& dZ(t)^*
\end{array}\right]=Q\left[\begin{array}{cc} 0 & 1  \\ 0 & 0
\end{array}\right]Q^Tdt.
$$
$H_{Int}$ can be implemented in one arm of a ring cavity with a fully reflecting mirror M and a partially
transmitting mirror M' with coupling coefficient $\gamma$, with $Z(t)$ incident on M'. After the interaction,
an output field $Z_{out}(t)$ is reflected by M' given by
\begin{eqnarray*}
Z_{out}(t)&=&U(t)^* Z(t) U(t)\\
&=&\frac{\tilde \alpha}{\sqrt{\gamma}} U(t)^*A(t)
U(t)-\frac{\tilde \beta}{\sqrt{\gamma}} U(t)^*A(t)^*U(t).
\end{eqnarray*}
However, the actual output that is of interest is the output $Y(t)=U(t)^*A(t)U(t)$ when the oscillator
interacts directly with the field $A(t)$. To recover $Y(t)$ from $Z_{out}(t)$, notice that since $Q$
is a quasi-unitary transformation, it has an inverse $Q^{-1}$ which is again quasiunitary. Hence $Y(t)$
can be recovered from $Z_{out}(t)$ by exploiting the following relation that follows directly from the
fact that $(Z_1(t),Z_1(t)^*)^T=Q(A_1(t),A_1(t)^*)^T$:
$$
\left[\begin{array}{c} Y(t) \\ Y(t)^*\end{array}\right]=Q^{-1}\left[\begin{array}{c} Z_{out}(t) \\
Z_{out}(t)^*\end{array}\right].
$$
That is, $Y(t)$ is the output of  a squeezer that implements the quasi-unitary transformation $Q^{-1}$
with $Z_{out}(t)$ as its input field. The complete implementation of this linear coupling is shown in
Figure~\ref{fig:diss-coup-2}.

\begin{figure}[tbp!]
\centerline{\includegraphics[width=13.5pc]{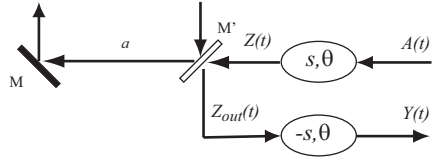}}
\caption{Scheme for implementation of a coupling $L=\tilde \alpha a+\tilde \beta a^*$ with
$\tilde \alpha>0$ and $\tilde \alpha> |\tilde \beta|$. Here $s=-{\rm  arccosh}(\tilde
\alpha/\sqrt{\gamma})$, $\theta=\arg(\tilde \beta)$ and the mirror M' has coupling coefficient
$ \gamma =\sqrt{\tilde \alpha^2 -|\tilde \beta|^2}$.\label{fig:diss-coup-2}}
\end{figure}

\subsection{Engineering the interactions between one-dimensional open quantum harmonic oscillators}\label{sec6.4}
The second necessary ingredient to synthesizing a general generalized open oscillator according to
Theorem \ref{th:synthesis} is to be able to implement a direct interaction Hamiltonian $H^d$ given by
(\ref{eq:multi-direct}) between one-dimensional harmonic oscillators. The only exception to this,
where field-mediated interactions suffice, is in the fortuitous instance where $R_{jk}$ and $L_j$ and
$S_j$, $j,k=1,\ldots,n$, are such that $H^d=0$. The Hamiltonian $H^d$ is essentially the sum of direct
interaction Hamiltonians between pairs of one-dimensional harmonic oscillators of the form
$H_{kl}=x_k^T C_{kl}x_l$ ($k \neq l$) with $C_{kl}$ a real matrix. Under the assumption that the time
it takes for the light in a ring cavity to make a round trip is much faster than the time scales of all
processes taking place in the ring cavity (i.e., the cavity length should not be too long), it will be
sufficient for us to only consider how to implement $H_{kl}$ for any two pairs of one-dimensional harmonic
oscillators and then implementing all of them simultaneously in a network. To this end, let
$a_j=(p_j + i q_j)/2$ and $a_j^*=(p_j-i q_j)/2$ for $j=k,l$, and rewrite $H_{kl}$ as
$$
H_{kl}=\epsilon_1 a_k^*a_l + \epsilon_1^* a_k a_l^*+
\epsilon_2 a_k^* a_l^* + \epsilon_2^* a_k a_l
$$
for some complex numbers $\epsilon_1$ and $\epsilon_2$. The first part
$H_{kl}^{1}= \epsilon_1 a_k^*a_l + \epsilon_1^* a_k a_l^*$ can be simply implemented by a beam splitter
with a mixing angle $\Theta=|\epsilon_1|$, $\Phi=-\arg(\epsilon_1)+\frac{\pi}{2}$, $\Psi=-\Phi$, and
$\Xi=0$ (see section \ref{sec:BS}). On the other hand, the second part
$H_{kl}^{2}= \epsilon_2 a_k^* a_l^* + \epsilon_2^* a_k a_l $ can be implemented by having the two modes
$a_k$ and $a_l$ interact in a suitable $\chi^{(2)}$ nonlinear crystal using a classical pump beam of
frequency $2\omega_r$ and effective pump intensity $-2i\epsilon_2$ in a two-mode squeezing process as
described in section \ref{sec:two-mode}. The overall Hamiltonian $H_{kl}$ can be achieved by positioning
the arms of the two ring cavities (with canonical operators $x_k$ and $x_l$) to allow their circulating
light beams to ``overlap'' at two points where a beam splitter and a nonlinear crystal are placed to
implement $H_{kl}^1$ and $H_{kl}^2$, respectively. An example of this is scheme is depicted in
Figure \ref{fig:cav-coup}.

\begin{figure}[tbp!]
\centerline{\includegraphics[width=12pc]{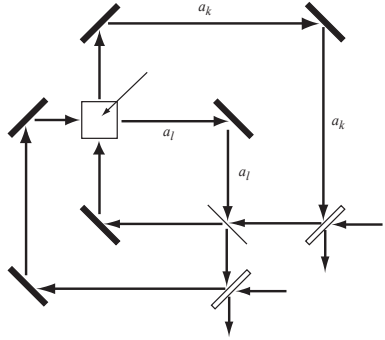}}
\caption{Example implementation of the total direct interaction $H_{kl}=H_{kl}^1+H_{kl}^2$ between
the modes $a_k$ and $a_l$ of two ring cavities.\label{fig:cav-coup}}
\end{figure}

\section{Illustrative synthesis example}\label{sec:example}
Consider a two degrees of freedom open oscillator $G$ coupled to a single external bosonic noise
field $A(t)$ given by $G=(I_{4 \times 4},Kx,x^T \diag(R_1,R_2)x)$, with $x=(q_1,p_1,q_2,p_2)^T$,
$K=[\begin{array}{cccc} 3/2 & 1/2i & 1 & i \end{array}]$,
\mbox{\scriptsize$R_1=\left[\begin{array}{@{}cc@{}} 2 & 0.5 \\
0.5 & 3 \end{array}\right]$}, and $R_2=\mbox{\scriptsize$\left[\begin{array}{@{}cc@{}} 1 & 0 \\ 0 & 1 \end{array}\right]$}$.

Let $G_1$ and $G_2$ be two independent one degree of freedom open oscillators given by
$G_1=(I_{2 \times 2},K_1 x_1,\half x_1^T R_{1} x_1)$ and
$G_2=(I_{2 \times 2},K_2x_2,\half x_2^T R_{2} x_2)$, with
$x_1=(q_1,p_1)^T$, $x_2=(q_2,p_2)^T$, $K_1=[\begin{array}{cc} 3/2
& i/2 \end{array}]$, and $K_2=[\begin{array}{cc} 1 & i\end{array}]$. Since the scattering matrix for
$G$ is an identity matrix, it follows from Theorem \ref{th:synthesis} that $G$ may be constructed as
a reducible network given by $G=\{\{G_1,G_2\}, H^d_{12},G_2 \triangleleft G_1 \}$ with the direct
interaction Hamiltonian $H^d_{12}$ between $G_1$ and $G_2$ given by (cf.\ (\ref{eq:multi-direct}))
\begin{eqnarray*}
H^d_{12}&=& -\frac{1}{2i} x_2^T (K_2^{\dag}K_1-K_2^T K_1^{\#})x_1\\
&=& \frac{1}{2}x_2^T \left[\begin{array}{cc} 0 & -1
\\ 3 & 0 \end{array}\right] x_1.
\end{eqnarray*}
This network is depicted in Figure \ref{fig:G1G2network}.

In the following we shall illustrate how to build $G_1$ and $G_2$ and how $H^d_{12}$ can be implemented
to synthesize the overall system $G$.

\subsection{Synthesis of \boldmath $G_1$ and $G_2$\unboldmath}\label{sec7.1}
Let us now consider the synthesis of $G_1=(I_{2 \times 2},K_1 x_1, \half x_1^T R_1 x_1)$. From the
discussion in section \ref{sec:one-degree-syn}, $R_1=\mbox{\scriptsize$\left[\begin{array}{@{}cc@{}} 2 & 0.5 \\ 0.5 & 3
\end{array}\right]$}$ can be realized as a DPA with parameters $\Delta=5$ and $\epsilon=1+i$, while
the coupling operator $L_1=K_1 x_1$ can be realized by the first scheme proposed in section
\ref{sec:one-degree-syn} and shown in Figure \ref{fig:diss-coup-1} by the combination of a two-mode
squeezer, a beam splitter, and an auxiliary cavity mode. Suppose that the coupling coefficient of
the mirror $M$ is $\gamma_2=100$; then the effective pump intensity of the two-mode squeezer is set
to be $10$ and the beam splitter should have a mixing angle of $-10$ with all other parameters
equal to $0$. Overall, the open oscillator $G_1$ with Hamiltonian $H_1=\half x_1^T R_1 x_1$ and
coupling operator $L_1$ can be implemented around a ring cavity structure, as shown in
Figure \ref{fig:G1}.  The open oscillator $G_2$ can be implemented in a similar way to $G_1$. The
Hamiltonian $H_2=\half x_2^T R_2 x_2$ can be implemented in the same way as $H_1$ with the choice
$\Delta=2$ and $\epsilon=0$. Since $\epsilon=0$, this means no optical crystal and pump beam are
required to implement $R_2$, but it suffices to have a cavity that is detuned from $\omega_r$, the
reference frequency in Remark \ref{rem:ref-frame}, by an amount $\Delta=2$. The coupling operator
$L_2=q_2+ip_2=2a_2$, where $a_2$ is the annihilation operator/cavity mode of cavity is standard
and can be implemented simply with a partially transmitting mirror with coupling coefficient
$\kappa=4$, on which an external vacuum noise field $A_2(t)$ interacts with the cavity mode $a_2$
to produce an outgoing field $Y_2(t)$. The implementation of $G_2$ is shown in
Figure~\ref{fig:G2}.

\begin{figure}[tbp!]
\centerline{\includegraphics[width=9pc]{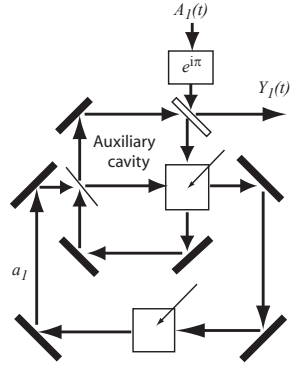}}
\caption{Realization of $G_1$.\label{fig:G1}}
\end{figure}

\begin{figure}[tbp!]
\centerline{\includegraphics[width=7.5pc]{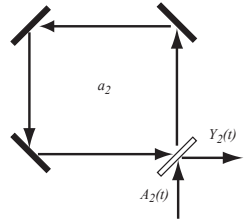}}
\caption{Realization of $G_2$.\label{fig:G2}}
\vspace{-9pt}
\end{figure}

\subsection{Synthesis of \boldmath $H^d_{12}$\unboldmath}\label{sec7.2}
We now consider the implementation of the direct interaction Hamiltonian $H_{12}^d$ given by
$H_{12}^d=\frac{1}{2}x_2^T \mbox{\scriptsize$\left[\begin{array}{@{}cc@{}} 0 & -1 \\ 3 & 0 \end{array}\right]$} x_1.$
To proceed, we first note that $H^d_{12}$ may be reexpressed in terms of the cavity modes $a_1$
and $a_2$ as $H^d_{12}=-i(a_1^* a_2-a_1a_2^*)+2i(a_1^* a_2^*-a_1a_2)$. Define
$H^d_{12,1}=-i(a_1^* a_2-a_1a_2^*)$ and $H^d_{12,2}=2i(a_1^* a_2^*-a_1 a_2)$ so that
$H^d_{12}=H^d_{12,1}+H^d_{12,2}$. The first part $H^d_{12,1}= -i(a_1^* a_2-a_1a_2^*)$ can be
simply implemented as a beam splitter with a rotation/mixing angle $\Theta=-1$ and all other
parameters set to $0$ (cf.\ section \ref{sec:BS}). On the other  hand, the second part
$H^d_{12,2}=2i(a_1^* a_2^*-a_1a_2)$ can be implemented by having the two modes $a_k$ and $a_l$
interact in a suitable $\chi^{(2)}$ nonlinear crystal using a classical pump beam of frequency
$\omega_p=2\omega_r$ and effective intensity $\epsilon=4$.

\subsection{Complete realization of \boldmath $G=\{\{G_1,G_2\},H_{12}^d,G_2\triangleleft G_1\}$\unboldmath}\label{sec7.3}
The overall two degrees of freedom open oscillator $G$ can now be realized by (i) positioning the
arms of the two (ring) cavities of $G_1$ and $G_2$ to allow their internal light beams to ``overlap''
at two points where a beam splitter and a nonlinear crystal are placed to implement $H_{12,1}^d$
and $H_{12,2}^d$, respectively, and (ii) passing the output $Y_1(t)$ of $G_1$ as input to $G_2$.
This implementation is shown in Figure \ref{fig:Grealization}.

\begin{figure}[tbp!]
\centerline{\includegraphics[width=17pc]{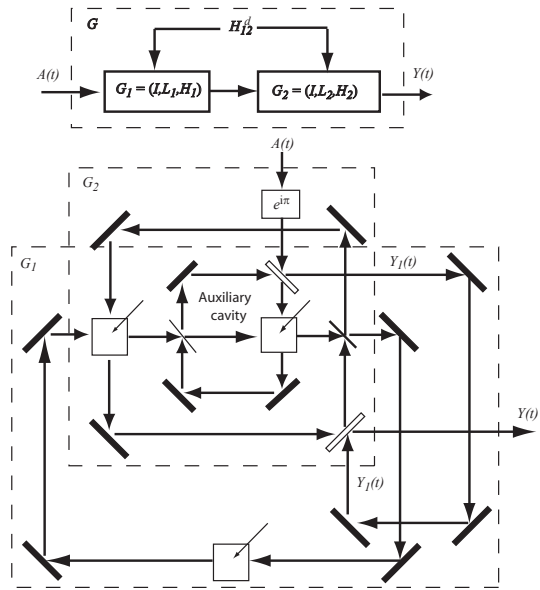}}
\caption{Realization of $G$. The block diagram at the top shows how $G$ is realized by a series
connection of $G_1$ into $G_2$ and a bilinear direct interaction $H^d_{12}$ between the canonical
operators of $G_1$ and $G_2$. The bottom figure shows the physical implementation of $G$ based on
the block diagram.\label{fig:Grealization}}
\end{figure}

\section{Conclusions}\label{sec:conclude}
In this paper we have developed a network theory for synthesizing arbitrarily complex linear dynamical
quantum stochastic systems from one degree of freedom open quantum harmonic oscillators in a systematic
way. We also propose schemes for building the one degree of freedom oscillators and the required
interconnections and interactions among them, in the setting of quantum optics.

Together with advances in experimental physics and the availability of high-quality basic quantum
devices, it is hoped the results of this work will assist in the construction of high-performance coherent
linear quantum stochastic controllers and linear photonic circuits in the laboratory for applications in
quantum control and quantum information science.

\appendix
\def\Hcal{\mathcal{H}}
\def\Fcal{\mathcal{N}}
\def\Dcal{\mathcal{D}}
\def\l2{l^2(\mathbb{Z}_+)}

\section*{Adiabatic elimination of coupled cavity modes}\label{sec:App-A}
In this section, we shall derive formulas for two coupled cavity modes in which one of the cavities has
very fast dynamics compared to the other and can be adiabatically eliminated, leaving only the dynamics
of the slow cavity mode. The cavities are each coupled to separate bosonic fields and are interacting
with one another in a classically pumped nonlinear crystal. A mathematically rigorous theory for the
type of adiabatic elimination/singular perturbation that we are interested in here has recently been
developed in \cite{BvHS07}.

The two cavity modes will be denoted by $a$ and $b$, each defined on two distinct copies of the
Hilbert space $\l2$ of square-integrable sequences ($\mathbb{Z}_+$ denotes the set of all nonnegative
integers). Thus the composite Hilbert space for the two cavity modes is $\Hcal=l^2(\mathbb{Z}_+)
\otimes l^2(\mathbb{Z}_+)$. The interaction in a nonlinear crystal is given, in some rotating frame,
by an interaction Hamiltonian $H_{ab}$ of the form $H_{ab}=\alpha a^*b+\beta a^*b^*+ \alpha^*
ab^*+\beta^* ab$ for some complex constants $\alpha$ and $\beta$. The mode $a$ is coupled to a
bosonic field $A_1$, while $b$ is coupled to the bosonic field $A_2$, both fields in the vacuum
state. The fields $A_1$ and $A_2$ live, respectively, on boson Fock spaces $\mathcal{F}_1$ and
$\mathcal{F}_2$, and we denote $\mathcal{F}=\mathcal{F}_1 \otimes \mathcal{F}_2$. We take $a$ to
be the slow mode to be retained and $b$ to be the fast mode to be eliminated.

We consider a sequence\vspace*{-1.5pt} of generalized open oscillators $G_k=(I,\tilde L^{(k)},H_{ab}^{(k)})$, with
$\tilde L^{(k)}=(\sqrt{\gamma_1} a,k\sqrt{\gamma_2}b)^T$ and
$H_{ab}^{(k)}=\Delta_1 a^*a + k^2\Delta_2 b^*b+ k(\alpha a^*b+\beta a^*b^*+ \alpha^* ab^*+\beta^* ab)$
each evolving according to the unitary $U_k$ satisfying the left H-P QSDE (as
opposed to the right  H-P QSDE in (\ref{eq:qsde-1})):
\begin{eqnarray*}
dU_k(t)&=&U_k(t)\bigg( \tilde L^{(k)\dag}  (dA_1(t),dA_2(t))^T-
\tilde L^{(k)T}
(dA_1(t),dA_2(t))^{\dag}+iH_{ab}^{(k)}\nonumber \\
&&\quad\hspace*{25pt} -\;\frac{1}{2}\tilde L^{(k){\dag}} \tilde
L^{(k)}dt\bigg).
\end{eqnarray*}
Here we are using the left QSDE following the convention used in \cite{BvHS07} (see Remark 2 therein)
so that the interaction picture dynamics of an operator $x$ is given by $x(t)=U_k(t)xU_k(t)^*$. We shall
use the results of \cite{BvHS07} to show, in a similar treatment  to section 3.2 therein, that in the
limit as $k \rightarrow \infty$:
\begin{eqnarray}
\lim_{k \rightarrow \infty}\sup_{0 \leq t \leq
T}\|U_k(t)^*\phi-U(t)^* \phi\|=0\;\; \forall \phi \in \Hcal_0
\otimes \mathcal{F} \label{eq:qsde-conv1}
\end{eqnarray}
for any fixed time $T>0$, where $\mathcal{H}_0$ is an appropriate Hilbert subspace of $\Hcal$ (to be
precisely specified in the next paragraph) for a limiting unitary $U(t)$ (again as a left H-P
QSDE) satisfying
\begin{eqnarray}
\label{eq:elim-eq1}
dU(t)&=&U(t)\left(\left(\frac{i2\Delta_2+\gamma_2}{i2\Delta_2-\gamma_2}-1\right)d\Lambda_{22}
+\sqrt{\gamma_1}a^*dA_1(t)-\sqrt{\gamma_1}adA_1(t)^* \right.\\
&&
\hspace*{32pt}\vphantom{\left(\frac{\gamma_2+i\Delta_2}{-\gamma_2+i\Delta_2}-1\right)}
-i\sqrt{\gamma_2}\left(i\Delta_2 - \frac{\gamma_2}{2} \right)^{-1}(\alpha
a^*+\beta^*
a)dA_2(t)\nonumber\\
&& \hspace*{34pt}+\;i\frac{2\sqrt{\gamma_2}}{i2\Delta_2-\gamma_2}(\alpha^* a
+ \beta a^*)dA_2(t)^*+ \left(i\Delta_1-\frac{\gamma_1}{2}\right)a^*a dt\nonumber\\
&&\hspace*{31pt} \left.
\vphantom{\left(\frac{\gamma_2+i\Delta_2}{-\gamma_2+i\Delta_2}\right)}
+\left(i\Delta_2-\frac{\gamma_2}{2}\right)^{-1}(\alpha a^*+\beta^*
a)(\alpha^*a+\beta a^*)dt \right) \nonumber
\end{eqnarray}
on $\Hcal_0 \otimes \mathcal{F}$. Note that (\ref{eq:elim-eq1}) is a left H-P QSDE corresponding to
the right form in section \ref{sec:models} by noting that we may write
\begin{eqnarray*}
\lefteqn{\left(i\Delta_1-\frac{\gamma_1}{2}\right)a^*a+
\left(i\Delta_2-\frac{\gamma_2}{2}\right)^{-1}(\alpha a^*+\beta^*
a)(\alpha^*a+\beta a^*)}\\
&=&i\left(\Delta_1
a^*a-\frac{\Delta_2}{\Delta_2^2+(\frac{\gamma_2}{2})^2}(\alpha
a^*+\beta^* a)(\alpha^*a+\beta a^*)\right)-\frac{1}{2}(\tilde
L_1^{\dag} \tilde L_1+\tilde L_2^{\dag} \tilde L_2),
\end{eqnarray*}
with $\tilde L_1=\sqrt{\gamma_1}a$ and $\tilde L_2=i\sqrt{\gamma_2}(-i\Delta_2 - \frac{\gamma_2}{2}
)^{-1}(\alpha^* a +\beta a^*)$. As such, it satisfies the H-P Condition 1 of \cite{BvHS07}.

Let $\phi_0,\phi_1,\ldots$ be the standard orthogonal bases of $\l2$, i.e., $\phi_l$ is an infinite
sequence (indexed starting from $0$) of complex numbers with all zeros except a $1$ in the $l$th
place. First, let us specify that $\Hcal_0=\l2 \otimes \Cbb \phi_0$;  this is the subspace of $\Hcal$
where the slow dynamics of the system will evolve. Next, we define a dense domain
$\Dcal={\rm span}\{\phi_j \otimes \phi_l;\,j,l=0,1,2,\ldots \}$ of $\Hcal$. The strategy is
to show that \cite[Assumptions 2--3]{BvHS07} are satisfied, from which the desired result will
follow from \cite[Theorem 11]{BvHS07}.

From the definition of $H_{ab}^{(k)}$, $\tilde L^{(k)}$ and $U_k$ given above, and we can define the
operators $Y,A,B,G_1,G_2$, and $W_{jl}\;(j,l=1,2)$ in \cite[Assumption 1]{BvHS07} as follows:
$Y=(i\Delta_2-\frac{\gamma_2}{2})b^*b,A=i(\alpha a^*b+ \alpha^* ab^* + \beta a ^*b^* + \beta^*
ab),B=(i\Delta_1-\frac{\gamma_1}{2}) a^*a, G_1=\sqrt{\gamma_1}a^*,G_2=0,F_1=0,F_2=\sqrt{\gamma_2}b^*,
W_{jl}=\delta_{jl}$. Then we can define the operators $K^{(k)}$, $L_j^{(k)}$ in this assumption as
\begin{eqnarray*}
K^{(k)}&=&k^2 Y+kA+B, \quad L_j^{(k)}=kF_j+G_j\;(j=1,2).
\end{eqnarray*}

Let $P_0$ be the projection operator to $\Hcal_{0}$. Let us now address Assumption 2 of \cite{BvHS07}. From our
definition of $\Hcal_0$, it is clear that we have that (a) $P_0\Dcal \subset \Dcal$. Any element
of $P_0 \Dcal$ is of the form $f \otimes \phi_{0}$ for some $f \in {\rm span}\{\phi_l;\,l=0,1,2,\ldots\}$;
therefore, since $Y=(i\Delta_2-\frac{\gamma_2}{2})b^*b$ and $b \phi_{0}=0$, we find that (b)
$YP_0 d=0\ \forall\ d \in \Dcal$. Define the operator $\tilde Y$ on $\Dcal$ defined by
$\tilde Y f \otimes \phi_0=0$ and $\tilde Y f \otimes \phi_l= l^{-1}(i\Delta_2-\frac{\gamma_2}{2})^{-1}
f \otimes \phi_l$ for $l=1,2,\ldots$ ($\tilde Y$ can then be defined to all of $\Dcal$ by linear
extension). From the definition of $Y$ and $\tilde Y$, it is easily inspected that (c1)
$Y\tilde Y f=\tilde Y Y f=P_1 f\ \forall\ f \in \Dcal$, where $P_1=I-P_0$ (i.e., the projection
onto the subspace of $\Hcal$ complementary to $\Hcal_0$). Moreover, because of the simple form of
$\tilde Y$, it is also readily inspected that (c2) $\tilde Y$ has an adjoint $\tilde Y^*$ with a
dense domain that contains $\Dcal$. Since $F_1=0$, we have that (d1) $F_1^* P_0=0$ on $\Dcal$, while
since $F_2^* f \otimes \phi_0 = \sqrt{\gamma_2}b f \otimes \phi_0 = 0\;\forall f \in \l2$, we also
have (d2) $F_2^* P_0 =0$ on $\Dcal$. Finally, from the expression for $A$ and the orthogonality of
the bases $\phi_0,\phi_1,\ldots$, a little algebra reveals that (e) $P_0 A P_0 d =0\ \forall\ d \in \Dcal$. From (a), (b), (c1--c2), (d1--d2), and (e),  we have now verified that Assumption 2
of \cite{BvHS07} is satisfied.

Finally, let us check that the limiting operator coefficients
$K,L_1,L_2,M_1, M_2$, and $N_{jk}\;(i,j=1,2)$ (as operators on $\Hcal_0$) of Assumption 3 of \cite{BvHS07} coincide with
the corresponding coefficients of (\ref{eq:elim-eq1}). These operator coefficients are defined as
$K=P_0(B-A \tilde Y A)P_0$, $L_j=P_0(G_j-A \tilde Y F_j)P_0$, $M_j=-\sum_{r=1}^{2}P_0 W_{jr}
(G_r^*-F_r^*\tilde Y A)P_0$, and $N_{jl}=\sum_{r=1}^{2}P_0 W_{jr}(F_r^*\tilde Y F_l +
\delta_{rl})P_0$. From these definitions and some straightforward algebra, we find that for all
$f \in {\rm span}\{\phi_l;\,l=0,1,2,\ldots\}$
\begin{eqnarray*}
K f \otimes \phi_0
&=&\left(\left(i\Delta_1-\frac{\gamma_1}{2}\right)a^*a+\left(i\Delta_2-\frac{\gamma_2}{2}\right)^{-1}(\alpha
a^*+\beta^* a)(\alpha^*a+\beta a^*) \right) f
\otimes \phi_0,\\
L_1 f \otimes \phi_0&=&\sqrt{\gamma_1}a^* f \otimes \phi_0,\\
L_2 f \otimes \phi_0 & = & -i\sqrt{\gamma_2}\left(i\Delta_2 -
\frac{\gamma_2}{2}\right)^{-1}(\alpha a^*+\beta^* a) f \otimes \phi_0,\\
M_1 f \otimes \phi_0 & = & -\sqrt{\gamma_1}a f \otimes \phi_0,\\
M_2 f \otimes \phi_0 & = &
\sqrt{\gamma_2}\left(i\Delta_2-\frac{\gamma_2}{2}
\right)^{-1}(\alpha^* a + \beta a^*) f \otimes \phi_0,
\end{eqnarray*}
and
\begin{eqnarray*}
&& N_{11}f \otimes \phi_0 = f \otimes \phi_0,\quad N_{12}f \otimes
\phi_0 =0,\quad N_{21}f \otimes \phi_0=0,\\
&& N_{22}f \otimes
\phi_0=\frac{\gamma_2+i2\Delta_2}{-\gamma_2+i2\Delta_2}f \otimes
\phi_0.
\end{eqnarray*}
Therefore, we see that $U(t)$ may be written as
$$
dU(t)=U(t)\left(\sum_{j,l=1}^{2}(N_{jl}-\delta_{jl})d\Lambda_{jl}
+\sum_{j=1}^{2}
 M_{j}dA_j^*+\sum_{j=1}^{2}L_j dA_j +K dt\right).
$$
Since we have already verified that (\ref{eq:elim-eq1}) is a bona fide right-QSDE equation, it now
follows that Assumption 3 of \cite{BvHS07} is satisfied. Now (\ref{eq:qsde-conv1}) follows from
\cite[Theorem 11]{BvHS07}, and the proof is complete.

Moreover, we can observe from the derivation above that the coupling of $a$ to $A_2(t)$ after
adiabatic elimination will not change if $a$ is also coupled to other cavities modes
$b_3,\ldots,b_m$ via an interaction Hamiltonian of the form
$\sum_{i=j}^{m}(\alpha_{j1} ab_j^*+ \alpha_{j1}^* a^*b_j + \alpha_{j2} a^*b_j^* + \alpha_{j2}^* a b_j)$,
and each additional mode may also linearly coupled to distinct bosonic fields $A_3,\ldots,A_m$,
respectively, as long as these other modes are {\em not} interacting with $b$ and with one another
(this amounts to just introducing additional operators $F_j,G_j$, $j \geq 3$, etc.). Moreover, under
these conditions one can also adiabatically eliminate any of the additional modes, and the only effect
will be the presence of additional sum terms in $U(t)$ that do not involve $b$, $A_1(t)$, and $A_2(t)$.

\def\Fock{\mathcal{F}(L^2(\mathbb{R}))}
\def\Rbb{ \mathbb{R} }
\def\Cbb{ \mathbb{C} }
\def\Fcal { \mathcal{F} }
\def\L2R{L^2(\Rbb)}

\section*{Squeezed white noise calculus}\label{sec:App-B}
The purpose of this appendix is to briefly recall results from the theory squeezed white noise
calculus \cite{HHKKR02} that are relevant as a basis for some formal calculations presented in
section \ref{sec:one-degree-syn}. As the theory is quite involved, it is not our intention here
to discuss any aspects of it in detail, but instead to point the reader to specific results of
\cite{HHKKR02}.

Let $\Fock$ denote the usual (symmetric) boson Fock space over the Hilbert space $L^2(\Rbb)$ of
complex-valued square integrable functions on $\Rbb$. Let $\Omega_{\Fcal}$ be the Fock vacuum
vector, and let $a_0(f)$ and $a_0^*(g)$ for $f,g \in \L2R$ be the vacuum creation and annihilation
operators on $\Fock$, respectively. Let $n \in \Rbb$ and $c\in \Cbb$ satisfy $n \geq 0$,
$n \geq 0$, and $n(n+1) \geq c$. The parameters $n,c$ characterize the so-called squeezed white
noise states $\omega_{n,c}$ \cite[section 2.1]{HHKKR02} that are postulated to satisfy the properties
(2.1)--(2.5) therein (see also \cite[Chapter 10]{GZ00}). However, here we will only be interested in
the special case of squeezed states with $n,c$ satisfying the constraint $n(n+1)=|c|^2$, as this is
the special case of squeezed states that can be generated from the vacuum state $\omega_0$ by an
appropriate squeezing Bogoliubov transformation \cite[equation (2.16)]{HHKKR02}; see \cite[Theorem  2.3]{HHKKR02}.
It has been shown that the annihilation and creation operators $a_{n,c}(f)$ and $a_{n,c}^*(g)$
($f,g \in \L2R$) corresponding to such a squeezed states can be concretely realized as operators on
$\Fock$ \cite[Theorem 2.11 part (b)]{HHKKR02} and  are given in terms of the vacuum creation and
annihilation operators $a_0(f)$ and $a_0(g)$, as (this follows from \cite[Theorem 2.3 and equation (3.12)]{HHKKR02})
\begin{eqnarray*}
a_{n,c}(f)&=&\cosh(s)a_0(f)+e^{i\theta}\sinh(s)a_0^*(Jf),\\
a_{n,c}^*(f)&=&\cosh(s)a_0^*(f)+e^{-i\theta}\sinh(s)a_0(Jf),
\end{eqnarray*}
where $J:f \mapsto f^*$, $s={\rm arctanh}(\frac{2|c|}{2n+1})$, and $\theta=\arg(c)$. Conversely, we have
$n=\frac{1}{2}\cosh(2s)-\frac{1}{2}$ and $c=\frac{1}{2}e^{i\theta}\sinh(2s)$. The squeezed white noise
state $\omega_{n,c}$ acts on an operator $A$ affiliated to the von Neumann algebra
$\Pi_{n,c}(\mathcal{W}(\L2R))''$ of operators on $\Fock$ (here $\Pi_{n,c}(\mathcal{W}(\L2R))$ denotes
the Gelfand--Naimark--Segal representation of the Weyl C*-algebra $\mathcal{W}(\L2R)$ on $\Fock$ corresponding to the state
$\omega_{n,c}$, and $''$ denotes the double commutant) as $$\omega_{n,c}(A)=\langle \Omega_{\Fcal},
A\Omega_{\Fcal}\rangle,$$ where $\langle \cdot,\cdot \rangle$ is the complex inner-product on $\Fock$
(antilinear in the first slot and linear in the second).

Let $A(t)=a(\chi_{[0,t]})$ be a vacuum bosonic field, where $\chi_{[0,t]}$ denotes the indicator
function for the interval $[0,t]$, and define the squeezed bosonic field
$A_{n,c}(t)=a_{n,c}(\chi_{[0,t]})$ with $a_{n,c}$ as defined above. Then $A_{n,c}$ and its adjoint
$A_{n,c}^*$ are related to $A$ and $A^*$ by
\begin{eqnarray}
\label{eq:App-1} A_{n,c}(t)&=&\cosh(s)A(t)+e^{i\theta}\sinh(s)A^*(t),\\
A_{n,c}^*(t)&=&\cosh(s)A^*(t)+e^{-i\theta}\sinh(s)A(t). \nonumber
\end{eqnarray}
Now, consider an open oscillator whose dynamics are given by the H-P QSDE:
\begin{eqnarray}
dU(t)&=& \left(-iH+dA(t)^{\dag}L - L^{\dag}dA(t) - \frac{1}{2} L^{\dag}L
dt \right)U(t), \label{eq:App-4}
\end{eqnarray}
where $H$ is the quadratic Hamiltonian of the oscillator and $L$ is the linear coupling operator to
$A(t)$. By using (\ref{eq:App-1}) and substituting this into the above QSDE, we may rewrite it in
terms of the $A_{n,c}$ and $A_{n,c}^*$ as follows:
\begin{eqnarray}
\label{eq:App-2} dU(t)&=& \bigg(-iH + dA_{n,c}(t)^{\dag} M -
M^{\dag}dA_{n,c}(t)  \\[-3pt]
&&\hspace*{10pt}-\;\frac{1}{2}(nMM^*+(n+1)M^*M-c^*M^2-cM^*M)dt\bigg)U(t),\nonumber
\end{eqnarray}
where $M$ is a new linear coupling operator given by
$$
M=\cosh(s)L+e^{i\theta}\sinh(s)L^*.
$$
As shown in \cite{HHKKR02}, (\ref{eq:App-2}) can be interpreted on its own as the unitary evolution
of a harmonic oscillator and a squeezed bosonic field linearly coupled via the coupling operator
$M$, and this defines a quantum Markov process on the oscillator algebra (by projecting to the
oscillator algebra; see \cite[section 3]{HHKKR02}). In this interpretation of (2), the squeezed bosonic
fields $A_{n,c}$ and $A_{n,c}^*$ satisfy the squeezed Ito multiplication rules given by\enlargethispage{6pt}
\begin{eqnarray*}
&&dA_{n,c}^2=cdt,\quad
dA_{n,c}dA_{n,c}^*=(n+1)dt,\quad
dA_{n,c}^*dA_{n,c}= n,\quad
(dA_{n,c}^*)^2=c^*dt,\\
&&dA_{n,c}dt=0,\quad dA_{n,c}^*dt=0
\end{eqnarray*}
that forms a basis for a quantum stochastic calculus for squeezed bosonic fields. A formal interpretation
of this is that (\ref{eq:App-2}) defines the evolution of a system coupled to $A_{n,c}$ via the formal
interaction Hamiltonian (see \cite[section 3.6]{HHKKR02}):
\begin{eqnarray}
H_{Int}(t)=i(M \eta_{n,c}^*(t)-M^*\eta_{n,c}(t)), \label{eq:App-3}
\end{eqnarray}
where $\eta_{n,c}$ is a squeezed quantum white noise that can be formally written as
$\eta_{n,c}=a_{n,c}(\delta(t)).$\footnote{As is often the case, there is technical caveat in that
for mathematical convenience the results of \cite{HHKKR02} are derived on the assumption that $H$
and $M$ are bounded operators on the oscillator Hilbert space. Here we do not concern ourselves too
much with such detail and assume the optimistic view that these results can be extended to unbounded
coupling operators $M$, which are linear combinations of the canonical operators of the harmonic
oscillator, in view of the fact that  the left form (cf.\ Appendix A) of (\ref{eq:App-4}), from which
the left form of (\ref{eq:App-2}) can be recovered, still makes sense for a quadratic $H$ and the
unbounded operator $L$ associated with $M$ (i.e., $L=\cosh(s)M-e^{-i\theta}\sinh(s)M^*$) \cite{Fagno90}.
Moreover, singular interaction Hamiltonians of the form (\ref{eq:App-3}) between the unbounded canonical
operators of a harmonic oscillator and a vacuum or squeezed quantum white noise are physically
well motivated and widely used in the physics community. See, e.g., \cite[Chapters 5 and 10]{GZ00}
and related references from \cite[section 3.6]{HHKKR02}.} The connection with the discussion in
section \ref{sec:one-degree-syn} is made by identifying the field $Z(t)$ introduced therein with
$A_{n,c}(t)$, and $\eta'(t)$ with $\eta_{n,c}(t)$.

\vspace*{-6pt}

\begin{thebibliography}{10}

\bibitem{Feyn59}
{\sc R.~Feynman},
{\em There's plenty of room at the bottom}, J. Microelectromech.\ Syst., 1 (1992),
pp.~60--66.

\bibitem{DM03}
{\sc J.~P. Dowling and G.~J. Milburn},
{\em Quantum technology: The second quantum revolution},
Philos. Trans. R. Soc. Lond. Ser.\ A Math. Phys. Eng. Sci., 361 (2003), pp.~1655--1674.

\bibitem{VPB83}
{\sc V.~Belavkin},
{\em On the theory of controlling observable quantum systems}, Automat. Remote Control,
44 (1983), pp.~178--188.

\bibitem{HW94a}
{\sc H.~Wiseman},
{\em Quantum theory of continuous feedback}, Phys.\ Rev.\ A (3), 49 (1994), pp.~2133--2150.

\bibitem{DJ99}
{\sc A.~Doherty and K.~Jacobs},
{\em Feedback-control of quantum systems using continuous state-estimation}, Phys.\ Rev.\ A (3),
60 (1999), pp.~2700--2711.

\bibitem{AASDM02}
{\sc M.~A. Armen, J.~K. Au, J.~K. Stockton, A.~C. Doherty, and H.~Mabuchi},
{\em Adaptive homodyne measurement of optical phase}, Phys.\ Rev.\ A (3), 89 (2002),
pp.~133602-1--\break133602-4.

\bibitem{SGDM04}
{\sc J.~K. Stockton, J.~M. Geremia, A.~C. Doherty, and H.~Mabuchi},
{\em Robust quantum parameter estimation: Coherent magnetometry with feedback}, Phys.\ Rev.\ A (3),
69 (2004), pp.~032109-1--032109-16.

\bibitem{Mab08}
{\sc H.~Mabuchi},
{\em Coherent-feedback quantum control with a dynamic compensator}, Phys.\ Rev.\ A (3), 78 (2008),
pp.~032323-1--032323--5.

\bibitem{Bel79}
{\sc V.~P. Belavkin},
{\em Optimal measurement and control in quantum dynamical systems},  Rep.\ Math.\ Phys., 43 (1999),
pp.~405--425.

\bibitem{EB05}
{\sc S.~C. Edwards and V.~P. Belavkin},
{\em Optimal quantum filtering and quantum feedback control}, 2005, http://arxiv.org/pdf/quant-ph/0506018.

\bibitem{YK03b}
{\sc M.~Yanagisawa and H.~Kimura},
{\em Transfer function approach to quantum control---part} II: {\em Control concepts and applications},
IEEE Trans.\ Automat. Control, 48 (2003), pp.~2121--2132.

\bibitem{JNP06}
{\sc M.~R. James, H.~I. Nurdin, and I.~R. Petersen},
{\em $H^{\infty}$ control of linear quantum stochastic systems}, IEEE Trans.\ Automat.\ Control,
53 (2008), pp.~1787--1803.

\bibitem{Nurd07}
{\sc H.~I. Nurdin},
{\em Topics in classical and quantum linear stochastic systems}, Ph.D.\ dissertation, The Australian
National University, Canberra, 2007.

\bibitem{NJP07a}
{\sc H.~I. Nurdin, M.~R. James, and I.~R. Petersen},
{\em Quantum {LQG} control with quantum mechanical controllers}, in Proceedings of the 17th IFAC World
Congress, Seoul, South Korea,  2008, pp.~9922--9927.

\bibitem{NJP07b}
{\sc H.~I. Nurdin, M.~R. James, and I.~R. Petersen},
{\em Coherent quantum {LQG} control}, Automatica, 45 (2009), pp.~1837--1846.

\bibitem{YD84}
{\sc B.~Yurke and J.~S. Denker},
{\em Quantum network theory}, Phys.\ Rev.\ A (3), 29 (1984), pp.~1419--1437.

\bibitem{TC06}
{\sc L.~Tian and S.~M. Carr},
{\em Scheme for teleportation between nanomechanical modes}, Phys.\ Rev.\ B, 74 (2006),
pp.~125314-1--124314--5.

\bibitem{SG08}
{\sc R.~J. Schoelkopf and S.~M. Girvin},
{\em Wiring up quantum systems}, Nature, 451 (2008), pp.~664--669.

\bibitem{AV73}
{\sc B.~D.~O. Anderson and S.~Vongpanitlerd},
{\em Network Analysis and Synthesis: A Modern Systems Theory Approach},  Networks Series, Prentice-Hall, Englewood Cliffs, NJ, 1973.

\bibitem{Leon03}
{\sc U.~Leonhardt},
{\em Quantum physics of simple optical instruments}, Rep.\ Prog.\ Phys., 66 (2003),
pp.~1207--1249.

\bibitem{GZ00}
{\sc C.~Gardiner and P.~Zoller},
{\em Quantum Noise: A Handbook of Markovian and Non-Markovian Quantum Stochastic Methods with
Applications to Quantum Optics}, Vol.~56, 2nd~ed.,  Springer Ser. Synergetics, Springer, New York, 2000.

\bibitem{HP84}
{\sc R.~L. Hudson and K.~R. Parthasarathy},
{\em Quantum Ito's formula and stochastic evolution}, Comm.\ Math.\ Phys., 93 (1984),
pp.~301--323.

\bibitem{KRP92}
{\sc K.~Parthasarathy},
{\em An Introduction to Quantum Stochastic Calculus}, Birkhauser, Berlin, 1992.

\bibitem{BvHJ07}
{\sc L.~Bouten, R.~{van Handel}, and M.~R. James},
{\em An introduction to quantum filtering}, SIAM J. Control Optim., 46 (2007), pp.~2199--2241.

\bibitem{YK03a}
{\sc M.~Yanagisawa and H.~Kimura},
{\em Transfer function approach to quantum control---part} I: {\em Dynamics of quantum feedback systems},
IEEE Trans.\ Automat. Control, 48 (2003), pp.~2107--2120.

\bibitem{Goug06a}
{\sc J.~Gough},
{\em Quantum {S}tratonovich calculus and the quantum {Wong-Zakai} theorem}, J. Math.\ Phys.,
47 (2006), pp.~113509-1--113509-19.

\bibitem{Fagno90}
{\sc F.~Fagnola},
{\em On quantum stochastic differential equations with unbounded coefficients},
Probab.\ Theory Related Fields, 86 (1990), pp.~501--516.

\bibitem{LN04}
{\sc U.~Leonhardt and A.~Neumaier},
{\em Explicit effective {Hamiltonians} for linear quantum-optical networks}, J. Opt.\ B Quantum Semiclass. Opt.,
6 (2004), pp.~L1--L4.

\bibitem{GJ07}
{\sc J.~Gough and M.~R. James},
{\em The series product and its application to quantum feedforward and feedback networks},
IEEE Trans. Automat. Control, 54 (2009), pp.~2530--2544.

\bibitem{GJ08}
{\sc J.~Gough and M.~R. James},
{\em Quantum feedback networks: {H}amiltonian formulation}, Comm.\ Math.\ Phys., 287 (2009),
pp.~1109--1132.

\bibitem{Nurd09b}
{\sc H.~I. Nurdin},
{\em Synthesis of linear quantum stochastic systems via quantum feedback networks}, IEEE Trans.
Automat. Control, to appear, 2010, http://arxiv.org/abs/0905.0802.

\bibitem{WM93}
{\sc H.~Wiseman and G.~Milburn},
{\em Quantum theory of field-quadrature measurements}, Phys.\ Rev.\ A (3), 47 (1993), pp.~642--663.

\bibitem{BvHS07}
{\sc L.~Bouten, R.~{van Handel}, and A.~Silberfarb},
{\em Approximation and limit theorems for quantum stochastic models with unbounded coefficients},
J. Funct. Anal., 254 (2008), pp.~3123--3147.

\bibitem{HHKKR02}
{\sc H.~Hellmich, R.~Honegger, C.~K{\"o}stler, B.~K{\"u}mmerer, and A.~Rieckers},
{\em Couplings to classical and non-classical squeezed white noise as stationary Markov processes},
Publ.\ Res. Inst. Math. Sci.,  38 (2002), pp.~1--31.
\end{thebibliography}
\end{document}